\documentclass{entcs} 
\usepackage{etex}
\usepackage{entcsmacro}
\usepackage{graphicx}
\usepackage{tikz}
\usepackage{verbatim}
\usepackage{cite}
\usepackage[all]{xy}
\usepackage{array}
\usepackage{tikz}
\usetikzlibrary{trees}
\usetikzlibrary{positioning}
\usepackage{slashbox}
\usepackage{ifsym}

\usepackage{fancyvrb}
\usepackage{epsfig}
\usepackage{epsf}
\usepackage{verbatim}
\usepackage{amssymb}
\usepackage{amsfonts,amsmath,latexsym}
\usepackage{bbm}
\usepackage{pgfplots}
\newcommand{\At}{\mathrm{\textrm{At}}}

\newcommand{\ls}{\mathcal{L}_{\Sigma}}

\def\lastname{ Komendantskaya, Schmidt, Heras}

\newtheorem{construction}{Construction}[section]
\begin{document}
\begin{frontmatter}
  \title{Exploiting Parallelism in Coalgebraic Logic Programming} 
\author[Dundee]{Ekaterina Komendantskaya\thanksref{coemail}\thanksref{ALL}}
\author[Osna]{Martin Schmidt\thanksref{myemail}}
 \author[Dundee]{J\'onathan Heras\thanksref{coemail2}\thanksref{ALL}}

\address[Dundee]{School of Computing\\University of Dundee\\
   UK} 
  \address[Osna]{Institute of Cognitive Science\\ University of Osnabr\"uck\\
    Germany} 
    \thanks[coemail]{Email:
    \href{mailto:katya@computing.dundee.ac.uk} {\texttt{\normalshape
        katya@computing.dundee.ac.uk}}}
    \thanks[myemail]{Email:
    \href{mailto:martisch@uos.de} {\texttt{\normalshape
        martisch@uos.de}}} 
    \thanks[coemail2]{Email:
    \href{mailto:jonathanheras@computing.dundee.ac.uk} {\texttt{\normalshape
        jonathanheras@computing.dundee.ac.uk}}}  
    \thanks[ALL]{The work was supported by EPSRC grants EP/J014222/1 and EP/K031864/1.}
        
\begin{abstract} 
We present a parallel implementation of Coalgebraic Logic Programming (CoALP) in the programming language Go. 
CoALP was initially introduced to reflect coalgebraic semantics of logic programming, 
with coalgebraic derivation algorithm featuring both corecursion and parallelism. 
Here, we discuss how the coalgebraic semantics influenced our parallel implementation of logic programming.
\end{abstract}


\end{frontmatter}

\section{Introduction}\label{sec:introduction}

In the standard formulations of Logic Programming (LP), such as in Lloyd's
book~\cite{Llo88}, a first-order logic program $P$ consists of a finite
set of clauses of the form
$A \leftarrow A_1, \ldots , A_n$,
where $A$ and the $A_i$'s are atomic first-order formulae, typically containing
free variables, and where $A_1, \ldots , A_n$ is understood to mean
the conjunction of the $A_i$'s: note that $n$ may be $0$.

SLD-resolution, which is a central algorithm for LP,
takes a goal $G$, typically written as
$ \gets B_1, \ldots , B_n$,
where the list of $B_i$'s is again understood to mean a conjunction of
atomic formulae, typically containing free variables, and constructs a
proof for an instantiation of $G$ from substitution instances of the
clauses in $P$~\cite{Llo88}.  The algorithm uses Horn-clause logic,
with variable substitution determined by most general unifiers to make a selected
atom in $G$ agree with the head of a clause in $P$, then proceeding
inductively.

Although the operational semantics of LP was initially given by the SLD-resolution algorithm, 
it was later reformulated in 
SOS style in \cite{AmatoLM09,BonchiM09,BruniMR01,CominiLM01,GLM95}, and in terms of  algebraic (fibrational) semantics in
\cite{AmatoLM09,BruniMR01,KP96,KPS12-2}.  
Logic programs resemble, and indeed induce, transition systems or
rewrite systems, hence coalgebras. That fact has been used to study
their operational semantics, e.g., in~\cite{BonchiM09,CominiLM01}. 
Finally, the coalgebraic (fibrational) semantics of LP was introduced in \cite{KMP10,KP11,KP11-2,KPS12-2,BonchiZ13}.
The main constructions and results of  \cite{KMP10,KP11,KP11-2,KPS12-2} will be explained in Sections \ref{sec:com} and \ref{sec:fib}.

When studying the coalgebraic (structural operational)  semantics of LP  \cite{KMP10,KP11,KP11-2,KPS12-2}, we noticed that some constructs of it  
suggest properties 
alien to the standard algorithm of SLD-resolution \cite{Llo88}; namely, parallelism and corecursion.
This paper will only focus on parallelism, but see \cite{KPS12-2} for a careful discussion of the relation between the two issues.
In particular, \cite{KMP10} first noticed  the relation of the coalgebraic semantics to parallel LP in the variable-free case \cite{GuptaC94}, as Section \ref{sec:com} explains.
However, extending those results to the first-order case with the fibrational coalgebraic semantics \cite{KP11,KPS12-2} again 
exposed novel constructions, this time alien to the existing  models of LP parallelism~\cite{GPACH12}. The ``fibers'' present
in it suggested restriction of the  
 unification algorithm standardly incorporated in 
SLD- and  and-or-parallel derivations \cite{GuptaC94} to term-matching; see Section~\ref{sec:fib}. 
This inspired us to introduce a new (parallel and corecursive) derivation algorithm of \emph{CoAlgebraic LP (CoALP)}  (see Section~\ref{sec:coder}).
The algorithm was shown sound and complete relative to the coalgebraic semantics \cite{KP11-2,KPS12-2}.

The original contribution of this paper is parallel implementation of CoALP in the language Go~\cite{S12}.
Go is a strongly typed and compiled programming language. 
It provides an easy built-in way to use high level constructs to implement parallelism in the form of \emph{goroutines} and channels to communicate between them;  
this model for providing high-level linguistic support for concurrency comes from Hoare's Communicating Sequential Processes \cite{Hoare:1978:CSP}. Go has an easy 
to set up and use tool-chain  and allows
for rapid prototyping with its fast compile times, array bounds checking and automatic memory management. In addition, it allows low level programming and produces fast binaries.

Here, we present a careful study of the influence of the constructs arising in the coalgebraic fibrational semantics \cite{KP11,KP11-2,KPS12-2} on\\
	(a) CoALP's parallel derivation algorithms;\\
	(b) the Go implementation of CoALP.\\
Thus, apart from achieving the goal of 
introducing CoALP\rq{}s implementation,
this paper
will serve as an exercise in applying coalgebra in programming languages.

The rest of the paper is structured as follows.   
Section \ref{sec:backgr} is a background section, and it introduces various existing (sequential and parallel) derivation 
algorithms for LP. 
The rest of the sections follow a common pattern: each splits into four subsections, such that 
the first subsection studies some constructions arising in the coalgebraic semantics \cite{KP11,KPS12-2}, the second subsection shows how 
those constructions transform into a parallel algorithm in CoALP; the third subsection explains their implementation in Go; and the last ``case study''
subsection tests the efficiency of their parallel implementation. 
Note the emphasis on describing the ``constructive'' fragments of the semantics, that is, fragments that give rise to concrete algorithms and computations.
In this way, Section~\ref{sec:par} considers propositional (variable-free) version of CoALP and related Datalog language. 
Section \ref{sec:cotrees} focuses on CoALP's fibrational semantics and its impact on parallelism. Section \ref{sec:coderivations} discusses semantics 
and parallelisation of full first-order fragment of CoALP. 
To reinforce the trend of tracing the constructive influence of the coalgebraic semantics on implementation of CoALP, we recover, where possible, a  
constructive reformulations of completeness results of \cite{KP11-2,KPS12-2}; and mark them as \lq\lq{}Constructive Completeness\rq\rq{} theorems/lemmas.
Where a constructive version is impossible, we discuss the reasons.
Finally, in Section \ref{sec:evaluation}, we conclude the paper.

The Go implementation of CoALP, as well as all the examples and benchmarks presented throughout the paper can be downloaded from~\cite{KPS13}.

\section{Background: Logic Programs and SLD-derivations}\label{sec:backgr}

We first recall some basic definitions from~\cite{Llo88}, and then proceed with discussion of parallel SLD-derivations.

  A \emph{signature} $\Sigma$ consists of a set of \emph{function
    symbols} $f,g, \ldots$ each equipped with a fixed
  \emph{arity}. The arity of a function symbol is a natural number
  indicating the number of its arguments.  Nullary (0-ary) function
  symbols are allowed: these are called \emph{constants}.
%
Given a countably infinite set $Var$ of variables, 
the set $Ter(\Sigma)$ of \emph{terms} over $\Sigma$ is defined
inductively:
$x \in Ter(\Sigma)$ for every $x \in Var$; and, 
if $f$ is an n-ary function symbol ($n\geq 0$) and $t_1,\ldots
  ,t_n \in Ter(\Sigma) $, then $f(t_1,\ldots
  ,t_n) \in Ter(\Sigma)$.  
Variables will be denoted $x,y,z$, sometimes with indices $x_1,
x_2, x_3, \ldots$.
%
A \emph{substitution} is a map $\theta : Ter(\Sigma) \rightarrow
Ter(\Sigma) $ which satisfies
$ \theta (f(t_1,\ldots
  ,t_n)) \equiv f(\theta (t_1),\ldots
  ,\theta(t_n))$
for every n-ary function symbol $f$.

We define an \emph{alphabet} to consist of a signature $\Sigma$, the set $Var$, and a set  
of \emph{predicate symbols} $P, P_1, P_2, \ldots$ each assigned an arity.
Let $P$ be a predicate symbol of arity $n$ and $t_1, \ldots, t_n$ be
terms. Then
$P(t_1, \ldots, t_n)$ is a \emph{formula} (also called an atomic formula or
an \emph{atom}).
The \emph{first-order language $\mathcal{L}$} given by an alphabet consists of
the
set of all formulae constructed from the symbols of the alphabet.

Given a substitution $\theta$ 
and an atom $A$, we write $A\theta$ for the atom
given by applying the substitution $\theta$ to the variables appearing
in  
$A$. Moreover, given a substitution $\theta$ and a list of atoms
$(A_1,\ldots ,A_k)$, we write $(A_1,\ldots,A_k)\theta$ for the simultaneous
substitution of $\theta$ in each $A_m$. 

  Given a first-order language $\mathcal{L}$, a \emph{logic program}
  consists of a finite set of clauses of the form $A \gets A_1, \ldots
  , A_n$, where $A, A_1, \ldots ,A_n$ ($n\geq 0$) are atoms.  The atom
  $A$ is called the \emph{head} of a clause, and $A_1, \ldots, A_n$ is
  called its \emph{body}.  Clauses with empty bodies are called
  \emph{unit clauses}.
We call a term, a formula, or a clause \emph{ground}, if it does not contain variables.

\begin{example}[\texttt{BinaryTree}]\label{ex:binarytree}
The definition \texttt{btree} describes a set of binary trees whose nodes are bits.

\begin{verbatim}
 1. bit(0). 
 2. bit(1).
 3. btree(empty).
 4. btree(tree(L,X,R)) :- btree(L), bit(X), btree(R).
\end{verbatim}

\end{example}

A \emph{goal} is given by $\gets B_1, \ldots B_n$, where $B_1, \ldots
B_n$ ($n\geq 0$) are atoms.

Let $S$ be a finite set of atoms. A substitution $\theta$ is called
  a \emph{unifier} for $S$ if, for any pair of atoms $A_1$ and $A_2$
  in $S$, applying the substitution $\theta$ yields $A_1\theta =
  A_2\theta$.  A unifier $\theta$ for $S$ is called a \emph{most
    general unifier} (mgu) for $S$ if, for each unifier $\sigma$ of
  $S$, there exists a substitution $\gamma$ such that $\sigma =
  \theta\gamma$. If $\theta$ is an mgu for $A_1$ and $A_2$, moreover, $A_1\theta = A_2$, then $\theta$ is a \emph{term-matcher}.


\begin{definition}\label{df:SLD}
Let a goal $G$ be $\gets A_1,\ldots ,A_m, \ldots, A_k$ and a clause $C$ be
$A\gets B_1, \ldots ,B_q$. Then $G'$ is \emph{derived} from $G$ and $C$ using
mgu
$\theta$ if the following conditions hold:
\begin{itemize}
\item[$\bullet$] $\theta$ is an \emph{mgu} of \textbf{the \emph{selected}}
atom $A_m$  in $G$ and $A$;
\item[$\bullet$] $G'$ is the goal $\gets(A_1, \ldots, A_{m-1},B_1, \ldots
,B_q, A_{m+1},
\ldots, A_k)\theta$.
\end{itemize}

\end{definition}

A clause $C^*_i$ is a \emph{variant} of the clause $C_i$ if $C^*_i =
C_i \theta$, with $\theta$ being a variable renaming substitution such
that variables in $C_i^*$ do not appear in the derivation up to
$G_{i-1}$. This process of renaming variables is called
\emph{standardising the variables apart}; we assume it throughout the
paper without explicit mention.

\begin{definition}\label{df:SLD2}
  An \emph{SLD-derivation} of $P\cup \{G\}$ consists of a sequence of
  goals $G=G_0, G_1, \ldots$ called \emph{resolvents}, a sequence
  $C_1,C_2, \ldots$ of variants of program clauses of $P$, and a
  sequence $\theta_1,\theta_2,\ldots$ of mgus such that each $G_{i+1}$
  is derived from $G_i$ and $C_{i+1}$ using $\theta_{i+1}$. An
  \emph{SLD-refutation} of $P \cup \{G\}$ is a finite SLD-derivation
  of $P \cup \{G\}$ that has the empty clause $\Box$ as its last
  goal. If $G_n = \Box$, we say that the refutation has length $n$. The composition $\theta_1, \theta_2, \ldots$ is called 
	\emph{computed answer}.
\end{definition}

Depending on the algorithm behind the choice of the ``selected atom'', and behind the choice of the program clause for a resolvent,
the proof-search strategy may differ. The most common strategy selects the left-most goal and top-most clause \cite{Llo88}. But the strategy may be changed to a
random choice, cf.~\cite{KPS12-2}. Also, there is an obvious choice between the breadth-first and depth-first search, if we view all the  SLD-choices as a tree.

\begin{example}
An SLD-derivation for the goal $\mathtt{btree(X)}$, with left-most atom, top-most clause and depth-first search is shown in the left side of Figure~\ref{fig:sld-derivation}.
Different strategies for that goal are represented in the right side of Figure~\ref{fig:sld-derivation}.

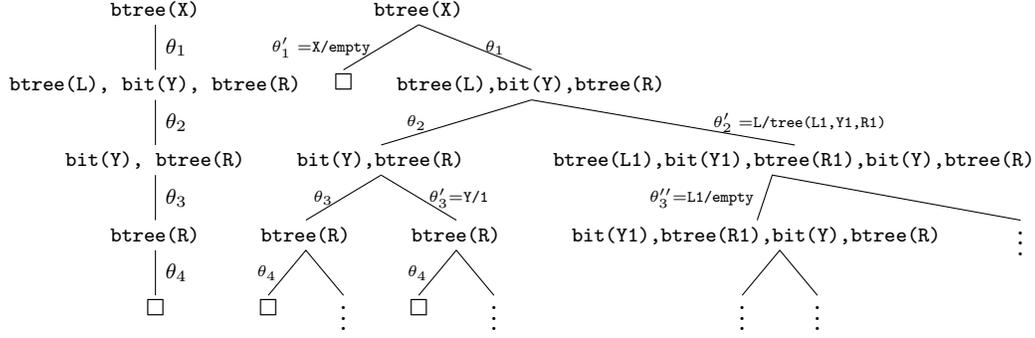
\begin{figure}
 \centering
 
 \begin{tikzpicture}
 \draw (-5,0) node {{\scriptsize \verb"btree(X)"}};
 \draw[] (-5,-0.2) -- node[anchor=west]{{\small $\theta_1$}} (-5,-.8);
 \draw (-5,-1) node {{\scriptsize \verb"btree(L), bit(Y), btree(R)"}};
  \draw[] (-5,-1.2) -- node[anchor=west]{{\small $\theta_2$}} (-5,-1.8);
 \draw (-5,-2) node {{\scriptsize \verb"bit(Y), btree(R)"}};
   \draw[] (-5,-2.2) -- node[anchor=west]{{\small $\theta_3$}} (-5,-2.8);
 \draw (-5,-3) node {{\scriptsize \verb"btree(R)"}};
    \draw[] (-5,-3.2) -- node[anchor=west]{{\small $\theta_4$}} (-5,-3.8);
 \draw (-5,-4) node {$\Box$};

   \draw (-1.5,0) node {{\scriptsize \verb"btree(X)"}};
   \draw (-1.5,-.2) --node[anchor=east]{{\tiny $\theta'_1=$\texttt{X/empty}}} (-2.5,-.8);
   \draw (-2.5,-1) node {$\Box$};
   
   \draw  (-1.5,-.2) --node[anchor=west]{{\tiny $\theta_1$}} (0,-.8);
   \draw (0,-1) node {{\scriptsize \verb"btree(L),bit(Y),btree(R)"}};
   
   \draw (0,-1.2) --node[anchor=east]{{\tiny $\theta_2~~~$}} (-2,-1.8);
   \draw (-2,-2) node {{\scriptsize \verb"bit(Y),btree(R)"}};
   
   \draw (0,-1.2) --node[anchor=west]{{\tiny $~~~~~~\theta'_2=$\texttt{L/tree(L1,Y1,R1)}}} (3.5,-1.8);
   \draw (3.5,-2) node {{\scriptsize \verb"btree(L1),bit(Y1),btree(R1),bit(Y),btree(R)"}};
   
   \draw (-2,-2.2) --node[anchor=east]{{\tiny $\theta_3$}} (-3,-2.8);
   \draw (-3,-3) node {{\scriptsize \verb"btree(R)"}};
   
   \draw (-3.5,-3.8) -- node[anchor=east]{{\tiny $\theta_4$}} (-3,-3.2) -- (-2.5, -3.8);
   \draw (-3.5,-4) node {$\Box$};
   \draw (-2.5,-4) node {$\vdots$};
   
      \draw (-2,-2.2) --node[anchor=west]{{\tiny $\theta'_3$=\texttt{Y/1}}} (-1,-2.8);
   \draw (-1,-3) node {{\scriptsize \verb"btree(R)"}};
      \draw (-1.5,-3.8) --node[anchor=east]{{\tiny $\theta_4$}} (-1,-3.2) -- (-.5, -3.8);
   \draw (-1.5,-4) node {$\Box$};
   \draw (-.5,-4) node {$\vdots$};
   
    \draw (3,-2.8) --node[anchor=east]{{\tiny $\theta''_3$=\texttt{L1/empty}}} (3.2,-2.2) -- (6.5,-2.8);
   \draw (3,-3) node {{\scriptsize \verb"bit(Y1),btree(R1),bit(Y),btree(R)"}};
     \draw (6.5,-3) node {\vdots};
     
    \draw (2.8,-3.8) --(3.3,-3.2) -- (3.8,-3.8);
    \draw (2.8,-4) node{\vdots};
    \draw (3.8,-4) node{\vdots};
 \end{tikzpicture}
 \caption{\textbf{Left.} An SLD-derivation (also a refutation) for \texttt{BinaryTree} with goal \texttt{btree(X)}. The computed answer is given by the composition of 
 $\theta_1=$ \texttt{X/tree(L,Y,R)}, $\theta_2=$ \texttt{L/empty}, $\theta_3=$ \texttt{Y/0}, $\theta_4=$ \texttt{R/empty}.
 \textbf{Right.} Different choices for SLD-derivations for the goal \texttt{btree(X)} selecting the left-most atom in a goal using a depth-first search 
 strategy.}\label{fig:sld-derivation}
\end{figure}
\end{example}

If we  pursue all possible selected atoms simultaneously instead of selecting one atom at a time, we will have an and-parallel implementation of the SLD-resolution. 
If we first pursue all possible clauses that unify with the given selected atom, we will have an or-parallel implementation. 
Pursuing both simultaneously gives and-or parallelism, see Figure \ref{fig:treeg}
and \cite{GuptaC94,PontelliG95,GPACH12}.


\begin{example}[\texttt{BinaryTree}]\label{ex:binarytree1}
{\normalsize 
The query {\verb"?- bit(X)"} can be solved simultaneously with {\verb"bit(0)"} and {\verb"bit(1)"} using or-parallelism. 
For the query { \verb"?- btree(tree(L,X,R))"}, an and-parallel algorithm can search for derivations for { \verb"btree(L)"}, {\verb"bit(X)"} 
and { \verb"btree(R)"} simultaneously.
}

\end{example}

\section{Parallel Derivations in Ground LP and Datalog}\label{sec:par}

We first discuss parallel derivation strategies in the ground  case.
Consider the ground re-formulation of the program \texttt{BinaryTree}.

\begin{example}[\texttt{BinaryTree - Ground Case}]\label{ex:binarytree2}
This ground logic program (let us call it \textbf{BTG}) defines a subset of the set of binary trees 
presented in Example \ref{ex:binarytree}.

\begin{verbatim}
 1. bit(0). 
 2. bit(1).
 3. btree(empty).
 4. btree(tree(empty,0,empty)) :- btree(empty), bit(0), btree(empty).
 5. btree(tree(empty,1,empty)) :- btree(empty), bit(1), btree(empty).
\end{verbatim}

\end{example}

\subsection{Coalgebraic Semantics for Derivations in LP}\label{sec:com}

Given a set $At$ of propositions (atoms), \cite{KPS12-2} shows that there is a bijection between the set of
variable-free logic programs over $At$ and the set of
$P_fP_f$-coalgebra structures on $At$, i.e., functions
$p:At\longrightarrow P_fP_f(At)$, where $P_f$ is the finite powerset
functor: each atom of a logic program $P$ is the head of finitely many
clauses, and the body of each of those clauses contains
finitely many atoms. 

The endofunctor $P_fP_f$ necessarily has a cofree comonad $C(P_fP_f)$ on
it.
It has been noticed in \cite{KMP10}, that, if a logic program can be modelled by a $P_fP_f$-coalgebra, then the SLD-derivations may be modelled
by a comonad $C(P_fP_f)$ on this coalgebra. 
The main result of~\cite{KMP10} established that, if $C(P_fP_f)$ is the
cofree comonad on $P_fP_f$, then, given a ground (variable-free) logic program $P$, the
induced $C(P_fP_f)$-coalgebra structure characterises the parallel and-or
derivation trees (cf. \cite{GuptaC94}) of $P$.
Here, we remind this construction, with a view of translating it first into a derivation algorithm, and then into implementation in Go.

\begin{example}\label{ex:pfpf}
Consider the logic program \textbf{BTG} from Example \ref{ex:binarytree2}.
The program has five atoms, namely \texttt{bit(0)}, \texttt{bit(1)}, \texttt{btree(empty)}, \texttt{btree(tree(empty,0,empty))} and \texttt{btree(tree(empty,1,empty))}. 
So $At_{BTG} = \{ \texttt{bit(0)}, \texttt{bit(1)}, \texttt{btree(empty)}, \texttt{btree(tree(empty,0,empty))},$\\
$ \texttt{btree(tree(empty,1,empty))}\}$. 
And the program can be identified with the $P_fP_f$-coalgebra structure on $At_{BTG}$ given by
$p(\texttt{bit(0)}) = \{ \{\empty\}\}$, $p(\texttt{bit(1)}) = \{ \{\empty\}\}$, $p(\texttt{btree(empty)}) = \{ \{\empty\}\}$ 
--- where $\{\empty \}$ is the empty set, and $ \{ \{\empty\}\}$, is the one element set consisting of the empty set ---
$p(\texttt{btree(tree(empty,0,empty)}) = \{\{\texttt{btree(empty), bit(0), btree(empty)}\}\}$,\\
$p(\texttt{btree(tree(empty,1,empty)}) = \{\{\texttt{btree(empty), bit(1), btree(empty)}\}\}$.
\end{example}

  Let $C(P_fP_f)$ denote the cofree comonad on $P_fP_f$. For any set
  $At$, $C(P_fP_f)(\At)$ is the limit of a diagram of the form
$$\ldots \longrightarrow \At \times P_fP_f(\At
\times P_fP_f(\At)) \longrightarrow \At \times P_fP_f(\At)
\longrightarrow \At.$$ Given $p:\At \longrightarrow P_fP_f(\At)$, put
$\At_0 = \At$ and $\At_{n+1} = \At \times P_fP_f(\At_n)$, and consider
the cone defined inductively as follows:
\begin{eqnarray*}
  p_0 & = & id: \At \longrightarrow \At\; ( = \At_0)\\
  p_{n+1} & = & \langle id, P_fP_f(p_n) \circ p \rangle : \At
  \longrightarrow \At \times P_fP_f (\At_n)\; ( = \At_{n+1})
\end{eqnarray*}
The limiting property determines the coalgebra
$\overline{p}: \At \longrightarrow C(P_fP_f)(\At)$.

\begin{example}\label{ex:comonad}
Continuing the previous example,
$\bar{p}(\texttt{btree(tree(empty,0,empty)}) )  =  p_2(\texttt{btree(tree(empty,0,empty)})) =  $ 
$ \langle \texttt{btree(tree(empty,0,empty))} \times \{ \{ \langle \texttt{btree(empty)} \times \{ \{ \empty \} \}\rangle , \langle \texttt{bit(0)} \times \{ \{ \empty\} \} \rangle , \langle \texttt{btree(empty)} \times \{ \{\empty\}\} \rangle \}\}\rangle $.
This construction could be graphically represented as a tree, see Figure~\ref{fig:comonad}.
If we think that every node of that tree is computed simultaneously and independently of the others, we may also say that 
Figure~\ref{fig:comonad} shows 
the and-or parallel SLD-refutation for the goal \texttt{btree(tree(empty,0,empty))}.


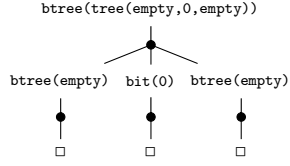
\begin{figure}
\begin{center}

\begin{tikzpicture}[level 1/.style={sibling distance=50mm},
level 2/.style={sibling distance=15mm},
level 3/.style={sibling distance=5mm},scale=.8,font=\tiny]
  
  \node (root) {\texttt{btree(tree(empty,0,empty))}} [level distance=6mm]
  child { [fill] circle (2pt)
             child { node {\texttt{btree(empty)}} 
                     child { [fill] circle (2pt)
                             child { node {$\Box$} }
                           }
                   }
             child { node {\texttt{bit(0)}} 
                     child { [fill] circle (2pt)
                             child { node {$\Box$} }
                           }
                   }
             child { node {\texttt{btree(empty)}} 
                    child { [fill] circle (2pt)
                             child { node {$\Box$} }
                           }
                   }
  }
     ;

%
%

  \end{tikzpicture}
\end{center}
\caption{\footnotesize{Graphical presentation of the action of $\overline{p}: \At \longrightarrow C(P_fP_f)(\At)$ on \texttt{btree(tree(empty,0,empty)}; this tree is also 
the and-or parallel derivation tree for \texttt{btree(tree(empty,0,empty)}.}}
\label{fig:comonad} 
\end{figure}

\end{example}

In Figure~\ref{fig:comonad}, the nodes alternate
between those labelled by atoms and those labelled by bullets
($\bullet$). Bullets correspond to the number of sets contained in the outer set. 
In Example~\ref{ex:comonad}, the big outer set contains one set with three elements, hence the tree root in 
Figure~\ref{fig:comonad} has one $\bullet$-node child, followed by further three children nodes. 
We use the traditional notation $\Box$ to denote
$\{\empty\}$. 

\subsection{From Semantics to Derivation Algorithm}

The following definition, first formulated in \cite{KMP10}, is CoALP's interpretation of and-or parallel derivations arising in Logic Programming, cf.~\cite{GPACH12}. 


\begin{definition}\label{df:andortree}
  Let $P$ be a ground logic program and let $G= \gets A$ be an atomic goal
  (possibly with variables). The \emph{and-or parallel derivation tree} for $A$ is the (possibly infinite) tree $T$ satisfying the
  following properties.
\begin{itemize}
\item $A$ is the root of $T$.
\item Each node in $T$ is either an and-node (an atom) or an or-node (given by $\bullet$).
\item For every node $A'$ occurring in $T$, if $A'$ is unifiable with exactly $m>0$ distinct clauses $C_1, \ldots , C_m$ in $P$ (a clause $C_i$ has the form $B_i
  \gets B^i_1, \ldots , B^i_{n_i}$, for some $n_i$) via mgu's $\theta_1,\ldots , \theta_m$, 
then $A'$ has exactly $m$ children given by or-nodes, such that, for every $i \in \{1, \ldots , m \}$, if $C_i  = B^i \gets B^i_1, \ldots ,B^i_n$, then the $i$th or-node has $n$ children 
given by and-nodes $B^i_1\theta_i, \ldots ,B^i_n\theta_i$.
\end{itemize}
\end{definition} 

\begin{example}
 An and-or tree corresponding to the \textbf{BTG} program with \verb"btree(X)" as goal is shown in Figure \ref{fig:treeg}.

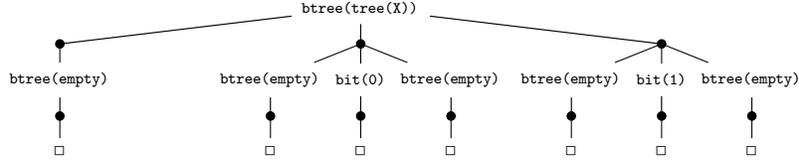
\begin{figure}
\begin{center}
\begin{tikzpicture}[level 1/.style={sibling distance=50mm},
level 2/.style={sibling distance=15mm},
level 3/.style={sibling distance=5mm},scale=.8,font=\tiny]
  
  \node (root) {\texttt{btree(tree(X))}} [level distance=6mm]
     child { [fill] circle (2pt)
             child { node {\texttt{btree(empty)}} 
                    child { [fill] circle (2pt)
                             child { node {$\Box$} }
                           }
                   }
           }
     child { [fill] circle (2pt)
             child { node {\texttt{btree(empty)}} 
                     child { [fill] circle (2pt)
                             child { node {$\Box$} }
                           }
                   }
             child { node {\texttt{bit(0)}} 
                     child { [fill] circle (2pt)
                             child { node {$\Box$} }
                           }
                   }
             child { node {\texttt{btree(empty)}} 
                    child { [fill] circle (2pt)
                             child { node {$\Box$} }
                           }
                   }
           }
     child { [fill] circle (2pt)
             child { node {\texttt{btree(empty)}} 
                     child { [fill] circle (2pt)
                             child { node {$\Box$} }
                           }
                   }
             child { node {\texttt{bit(1)}} 
                     child { [fill] circle (2pt)
                             child { node {$\Box$} }
                           }
                   }
             child { node {\texttt{btree(empty)}} 
                    child { [fill] circle (2pt)
                             child { node {$\Box$} }
                           }
                   }
           }
   ;

  \end{tikzpicture}
\end{center}
\caption{\footnotesize{The and-or parallel derivation tree for $\mathtt{btree(X)}$ for the  \texttt{BTG} program.}}
\label{fig:treeg} 
\end{figure}

\end{example}

The following ``Constructive Completeness\rq\rq{} result is a re-formulation of the more general soundness and completeness results of \cite{KMP10,KPS12-2}. The below formulation serves our ultimate goal of tracing the inheritance of constructions from coalgebraic semantics to logic algorithm and data structures used in implementation.

\begin{theorem}[Constructive Completeness]\label{th:CC}
Let $P$ be a ground logic program, and $G$ be a ground atomic goal. Given the construction of $\overline{p}(G)$,
there exists (can be constructed) an and-or tree $T_G$ for $G$, such that:
\begin{itemize}
\item (Tree depth 0.) The root of $T_G$ is given by $p_0(G) = G$.
\item (Tree depth $n$, for odd $n$.) Every node $A$ appearing at the tree depth $n-1$ has $m$ $\bullet$-child-nodes at the tree depth $n$, corresponding to the number of sets contained in the set $p(A)$.
\item (Tree depth $n$, for even $n>0$.) Every $i$th $\bullet$-node at the depth $n-1$ with a parent node $A$ at the level $n-2$ has children at the depth $n$, given by the distinct elements of the $i$th set contained in the set $p(A)$.
\end{itemize}
Moreover,  $T_G$  has finite depth $2n$ (for some $n \in \mathbb{N}$)
   iff ${\bar p}(G) = p_n(G)$. 
 The $T_G$ is infinite iff  ${\bar p}(G)$ is given
   by the element of the limit $\lim_{\omega}(p_n)(At)$ of an infinite
   chain given by the construction of $C(P_fP_f)$ above.   
\end{theorem}


\subsection{From Derivation Algorithm to Implementation}\label{subsec:and-or-g}

The above definition of the and-or-tree can give rise to an interpreter. 
CoALP's implementation in Go starts with the construction of template trees called \emph{clause-trees} -- they are generated 
from the input program and are the building blocks that will be used for the construction of and-or parallel derivation trees.


\begin{definition}\label{def:clause-tree}

Let $P$ be a ground logic program and let $C=A \gets  B_1, \ldots, B_n$ be a clause in $P$. The \emph{clause-tree} for $C$ is the tree $T$ satisfying the following properties:
 \begin{itemize}
 \item $A$ is the root of $T$.
 \item Each node in $T$ is either an and-node (an atom) or an or-node (given by $\bullet$). 
 \item $A$'s child is given by an or-node. This or-node has $n$ children given by and-nodes $B_1,\ldots,B_n$.
 \item For every clause $C'$ occurring in $P$ and for every and-node $B_i$, if $B_i$ is unifiable with the head of $C'$, then the node $B_i$ contains a reference to the root of the clause-tree of $C'$. 
 \end{itemize}

\end{definition}

Note that the first three items in the definition of the clause-tree mimic very closely the action of $P_fP_f$-coalgebra $p$ on elements of $At$. Whereas the last item of the above definition 
paves the way for implementation of the $\bar{p}$ construction. The next example makes the connection clear.

\begin{example}
The clause-tree structures with open list references for the \texttt{BTG} program are shown in Figure \ref{pic:clausetree-ground}.
Compare with Examples~\ref{ex:pfpf} and \ref{ex:comonad}.

\begin{figure}
\begin{center}
\footnotesize{
  \begin{tikzpicture}[scale=0.2,font=\footnotesize,baseline=(current bounding box.north),grow=down,level distance=30mm]
      \node (btree0) {$\mathtt{btree(tree(empty,0,empty))}$}
        child {[fill,sibling distance=95mm] circle (4pt)
            child { node (btree0L) {$\mathtt{btree(empty)}$} }
            child { node (btree0X) {$\mathtt{bit(0)}$} }
            child { node (btree0R) {$\mathtt{btree(empty)}$} } };
      \node [right of=btree0,node distance=4cm] (btreeE) {$\mathtt{btree(empty)}$}
        child {[fill] circle (4pt)
            child { node {\scriptsize $\Box$} } };
      \node [below of=btree0,node distance=2cm] (bit0) {$\mathtt{bit(0)}$}
         child {[fill,sibling distance=90mm] circle (4pt)
         child { node {\scriptsize $\Box$} } };
      \node [right of=btreeE,node distance=4cm] (btree1) {$\mathtt{btree(tree(empty,1,empty))}$}
                child {[fill,sibling distance=95mm] circle (4pt)
                    child { node (btree1L) {$\mathtt{btree(empty)}$} }
                    child { node (btree1X) {$\mathtt{bit(1)}$} }
                    child { node (btree1R) {$\mathtt{btree(empty)}$} } };
      \node [below of=btree1,node distance=2cm] (bit1) {$\mathtt{bit(1)}$}
        child {[fill,sibling distance=90mm] circle (4pt)
            child { node {\scriptsize $\Box$} } };
      
      \draw[->,dashed,shorten >=2pt]
          (btree0X) to[out=270,in=90] (bit0);
      \draw[->,dashed,shorten >=2pt]
          (btree0L) to[out=270,in=270] (btreeE);
      \draw[->,dashed,shorten >=2pt]
          (btree0R) to[out=270,in=270] (btreeE);
      \draw[->,dashed,shorten >=2pt]
          (btree1X) to[out=270,in=90] (bit1);
      \draw[->,dashed,shorten >=2pt]
          (btree1L) to[out=270,in=270] (btreeE);
      \draw[->,dashed,shorten >=2pt]
          (btree1R) to[out=270,in=270] (btreeE);
     \end{tikzpicture}
}
\end{center}
\caption{\footnotesize{Clause-trees for the \texttt{BTG} program with dashed lines denoting references between trees.}}
\label{pic:clausetree-ground} 
\end{figure}
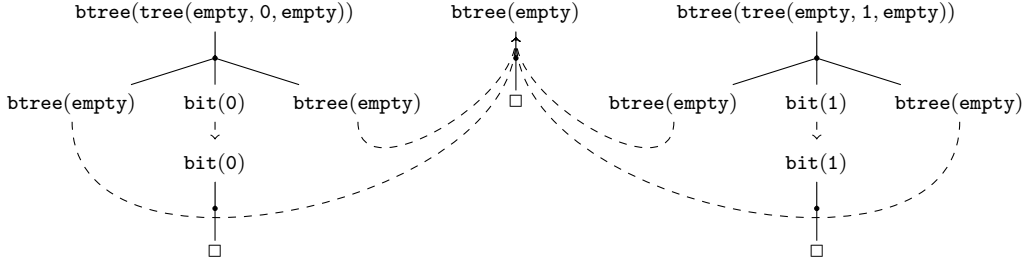

\end{example}


CoALP's implementation parses and transforms each program clause into a clause-tree when the program is loaded. We use two kinds of structures to encode clause-trees: 
and-nodes and or-nodes. The or-node structure consists of a list of pointers to and-nodes. The and-node structure consists of a clause head, together with a list of or-nodes, 
and a list to track references to clause-tree root nodes called \emph{open list}. Open lists play an important role in our implementation: they are used in a lazy fashion
to add new or-nodes to the or-node list in the future. 
Or-nodes and and-nodes are linked by pointers and are allocated dynamically, see Figure \ref{pic:clausetree-ground}. 

The construction of clause-trees from a given program consists of two steps. In the first one, the clauses of the program are transformed into clause-trees with
an empty list of references in the and-nodes. After the transformation pass, each and-node corresponding to a clause body atom is visited again and its open list is 
populated with references to the 
unifiable clause-tree root nodes. 
This is a one time process at the initialisation and does not need to be done again for different queries.

\begin{construction}[Derivations by Clause trees]\label{cons:and-or}
Given a program $P$ and a goal atom $G=\gets A$, a clause-tree derivation proceeds  to construct the tree $T$, as follows:

\begin{enumerate}
 \item A root $A$ for $T$ is created as an and-node containing the goal atom.
 \item The open list of the root $A$ is constructed by adding references to all clause-trees that have the same root atom.
 \item For each reference in an open list $O$ of a node $A'$ (where the corresponding atom equals the referenced root node's atom), a copy of the or-node below the referenced node and all its children in the clause-tree are added as a child to $A'$. The reference is then deleted from $O$.
 \item This process continues until all references in all the open lists in the tree $T$ have been processed.
\end{enumerate}
\end{construction}

\begin{example}\label{ex:ground-query-ct}
 
Given the query \texttt{btree(tree(empty,0,empty))} in the \textbf{BTG} program, we construct the and-or parallel tree as follows. We start with a tree only consisting of the goal 
atom as root and-node. The reference to the clause-tree with root \texttt{btree(tree(empty,0,empty))} is added to this
and-node open list as the respective atoms of the nodes are equal. Then, we start processing all nodes that have references to other clause-tree roots. 
We check whether the clause-tree root node that is referenced equals the currently processed nodes atom. If the equality has been verified, we can substitute the
node in our tree with a copy of the referenced clause-tree. In this case after the initial root node is expanded this process is done for the two leave nodes
\texttt{btree(empty)} and \texttt{bit(0)}. We continue this match and copy process until no nodes with references that match are left in our constructed tree.
For our query \texttt{btree(tree(empty,0,empty))}, the resulting tree will look like the and-or 
parallel tree depicted in Figure \ref{fig:comonad}, compare also with Example \ref{ex:comonad}. 

\end{example}

\begin{lemma}\label{lemma:and-or}
Let $P$ be a ground logic program and  $G$ be a ground atomic goal. Then, the and-or parallel derivation tree for $G$ 
is given by Construction~\ref{cons:and-or}.
\end{lemma}

Constructive Completeness of Theorem~\ref{th:CC} and the lemma above show the full chain starting from coalgebra and ending with implementation. It now remains to show that the resulting parallel language is indeed efficient.

Construction~\ref{cons:and-or} permits parallelisation since no variable  synchronization 
is required, and the order in which nodes of the tree are expanded is not relevant. Different expansion strategies ranging from sequential-depth and breadth-first up to fully parallel can be considered. This process in principle scales to the number of references to other clause-trees that the given tree has. 

During the construction of the parallel and-or tree, it has to be checked whether it contains a subset of branches that constitute
a proof for the query that is the root node of the tree. This work can be integrated into the expansion process and large parts can be done in parallel here as well. 

Note that implementing the above  restricted (ground) logic programs can have practical value of its own. Logic programs containing variables but no function symbols of arity $n>0$ can all be soundly translated into finitely-presented ground logic programs. The most famous example of such a language is Datalog~\cite{UllmanG88,Kanellakis88}.
The advantages of Datalog are easier implementations and a greater capacity for parallelisation. From the point of view of model theory, Datalog programs always have finite models. 

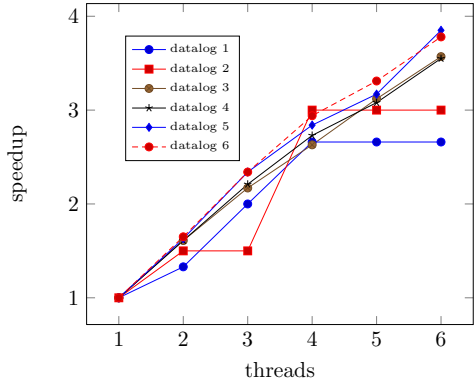
\begin{figure}
\centering
\begin{tikzpicture}[scale=.75]
	\begin{axis}[xlabel= threads,
		ylabel=speedup,legend style={at={(0.1,.9)},anchor=north west}]
	\addplot  coordinates {
		(1,1)
		(2,1.33)
		(3,2)
		(4,2.66)
		(5,2.66)
		(6,2.66)
	
	};
	\addplot coordinates {
		(1,1)
		(2,1.5)
		(3,1.5)
		(4,3)
		(5,3)
		(6,3)
	
	};
	\addplot coordinates {
		(1,1)
		(2,1.61)
		(3,2.17)
		(4,2.63)
		(5,3.12)
		(6,3.57)

	};
	\addplot  coordinates {
		(1,1)
		(2,1.61)
		(3,2.21)
		(4,2.73)
		(5,3.08)
		(6,3.55)
	
	};
	\addplot  coordinates {
		(1,1)
		(2,1.63)
		(3,2.34)
		(4,2.84)
		(5,3.17)
		(6,3.85)
	
	};
        \addplot  coordinates {
		(1,1)
		(2,1.65)
		(3,2.34)
		(4,2.94)
		(5,3.31)
		(6,3.78)
	
	};

	\legend{{\tiny datalog 1}, {\tiny datalog 2}, {\tiny datalog 3}, {\tiny datalog 4}, {\tiny datalog 5}, {\tiny datalog 6}}

	\end{axis}
\end{tikzpicture}
\caption{Speedup of Datalog programs, relative to the base case with 1 thread,  with different number of threads expanding the derivation tree.}\label{fig:datalog}
\end{figure}

Figure \ref{fig:datalog} shows the speedup that can be gained by constructing and-or parallel trees for Datalog programs in our system.
The Datalog programs are randomly generated and can be examined in~\cite{KPS13}. As can be seen in Figure~\ref{fig:datalog}, the speedup is
 significant and scales with the number of threads.

\subsection{DataLog Case Study: BTG}\label{subsec:dcsbtg}

 
Generally, given a query \texttt{btree(X)} and a ground variant of the \texttt{BinaryTree} program,  we take the matching ground instances of \texttt{X} to construct 
its subtrees (cf.~Figure~\ref{fig:treeg}).
 In the BTG fragment shown in Example \ref{ex:binarytree2}, these are
\texttt{btree(tree(empty,0,empty))},
\texttt{btree(tree(empty,1,empty))} and \texttt{btree(empty)}; but in some tests we describe here, there will be hundreds and even thousands of such instances. 
For each of these instances, parallel and-or trees can be constructed independently and in parallel. 
%
Our implementation provides several options to 
configure:

\begin{itemize}
 \item \emph{Number of threads}. This parameter indicates the number of threads that will be used in the general processing of the and-or parallel trees. 
 \item \emph{Parallel Expansion}. This option indicates that the expansion process will be run in parallel.
 \item \emph{Number of expansion threads}. This is the number of total additional threads that will be used to help expand the and-or trees in parallel.
\end{itemize}

In order to test our parallel implementation and emulate the real-life database growth in Datalog,  
we have increased the number of clauses of the
\texttt{BTG} program in two different ways. In Experiment 1, we use an algorithm (call it BTA) to generate hundreds of ground instances of clauses in \texttt{BinaryTree}, but making sure that the generated clauses describe balanced trees, i.e. having all the branches of the same depth.
In the second experiment, we do not impose this restriction, and use an algorithm (call it UTA) that generates ground programs describing hundreds of (balanced and unbalanced) binary trees;  we refer to the second kind of data as 
``unbalanced trees\rq\rq{}.
The algorithms BTA and UTA  
are given in Appendix~\ref{appendix1}.

There are various parameters to measure the success of a parallel language; we focus on three aspects, as follows.

\textbf{1. Program speedup with the increase of parallel threads;} or, in other words, \emph{given a program, would its parallel execution bring significant speedup?}\\
Figure~\ref{fig:speedup1}  shows the speedup when increasing the number of threads and expansion threads, for ten BTG Datalog programs, of different sizes and nature.  It also shows that the and-or derivation trees are generated faster when the \emph{parallel expansion} option 
is activated.
If the parallel expansion option is not used, increasing the number of threads does not significantly speed up the execution time 
(maximum speedup of $1.21$); on the contrary, using the parallel expansion option and increasing the number of expansion threads considerably speeds up the execution time
(maximum speedup of $4.13$).

\textbf{2. The gap between the best case and worst case of parallelisation:} or \emph{would any program be suitable for parallelisation?}

The BTA algorithm 
defines $2^{2^n-1}$ binary trees for a given depth $n$. Each of these trees will produce $2^{n+1}-1$ leaf nodes 
in the corresponding parallel and-or derivation tree, where $n$ is the depth of the tree.
The  UTA algorithm 
generates $3 \cdot 6^{2 \cdot n -3} - 3 \cdot 6^{2 \cdot n-5}$ trees of depth $n$.  
E.g., for depth $3$, UTA generates $630$ binary trees,
and BTA just $128$. This means that the programs created with the UTA algorithm contain smaller trees (regarding number of nodes), and 
this impacts the speedup, as it is shown in the right diagram of Figure~\ref{fig:speedup1}.



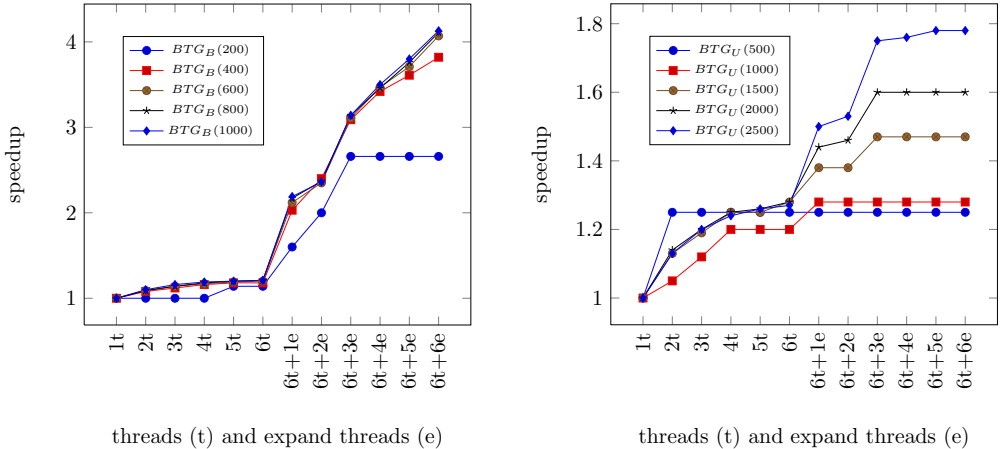
\begin{figure}
\centering
\begin{tikzpicture}

\draw (0,0) node{
\begin{tikzpicture}[scale=.75]
	\begin{axis}[xlabel= threads (t) and expand threads (e), xlabel style={at={(0.5,-.2)},anchor=north},
		ylabel=speedup,legend style={at={(0.1,.9)},anchor=north west},
		xticklabels={1t,2t,3t,4t,5t,6t,6t+1e,6t+2e,6t+3e,6t+4e,6t+5e,6t+6e},xtick={1,...,12}, x tick label style={rotate=90,anchor=east}]
	\addplot  coordinates {
		(1,1)
		(2,1)
		(3,1)
		(4,1)
		(5,1.14)
		(6,1.14)
		(7,1.6)
		(8,2)
		(9,2.66)
		(10,2.66)
		(11,2.66)
		(12,2.66)
	};
	\addplot coordinates {
		(1,1)
		(2,1.08)
		(3,1.12)
		(4,1.16)
		(5,1.18)
		(6,1.18)
		(7,2.03)
		(8,2.4)
		(9,3.09)
		(10,3.42)
		(11,3.61)
		(12,3.82)
	};
	\addplot coordinates {
		(1,1)
		(2,1.09)
		(3,1.14)
		(4,1.17)
		(5,1.19)
		(6,1.2)
		(7,2.12)
		(8,2.35)
		(9,3.11)
		(10,3.47)
		(11,3.71)
		(12,4.07)
				
	};
	\addplot  coordinates {
		(1,1)
		(2,1.09)
		(3,1.14)
		(4,1.18)
		(5,1.2)
		(6,1.21)
		(7,2.17)
		(8,2.37)
		(9,3.14)
		(10,3.46)
		(11,3.76)
		(12,4.11)
	};
	\addplot  coordinates {
		(1,1)
		(2,1.1)
		(3,1.16)
		(4,1.19)
		(5,1.2)
		(6,1.21)
		(7,2.19)
		(8,2.36)
		(9,3.14)
		(10,3.5)
		(11,3.8)
		(12,4.13)
	};

	\legend{{\tiny $BTG_B(200)$}, {\tiny $BTG_B(400)$}, {\tiny $BTG_B(600)$}, {\tiny $BTG_B(800)$},{\tiny $BTG_B(1000)$}}

	\end{axis}
\end{tikzpicture}};

\draw (7,0) node{
\begin{tikzpicture}[scale=.75]
	\begin{axis}[xlabel= threads (t) and expand threads (e), xlabel style={at={(0.5,-.2)},anchor=north},
		ylabel=speedup,legend style={at={(0.1,.9)},anchor=north west},
		xticklabels={1t,2t,3t,4t,5t,6t,6t+1e,6t+2e,6t+3e,6t+4e,6t+5e,6t+6e},xtick={1,...,12}, x tick label style={rotate=90,anchor=east}]
	\addplot  coordinates {
		(1,1)
		(2,1.25)
		(3,1.25)
		(4,1.25)
		(5,1.25)
		(6,1.25)
		(7,1.25)
		(8,1.25)
		(9,1.25)
		(10,1.25)
		(11,1.25)
		(12,1.25)
	};
	\addplot coordinates {
		(1,1)
		(2,1.05)
		(3,1.12)
		(4,1.2)
		(5,1.2)
		(6,1.2)
		(7,1.28)
		(8,1.28)
		(9,1.28)
		(10,1.28)
		(11,1.28)
		(12,1.28)
	};
	\addplot coordinates {
		(1,1)
		(2,1.13)
		(3,1.19)
		(4,1.25)
		(5,1.25)
		(6,1.28)
		(7,1.38)
		(8,1.38)
		(9,1.47)
		(10,1.47)
		(11,1.47)
		(12,1.47)
				
	};
	\addplot  coordinates {
		(1,1)
		(2,1.14)
		(3,1.2)
		(4,1.25)
		(5,1.26)
		(6,1.28)
		(7,1.44)
		(8,1.46)
		(9,1.6)
		(10,1.6)
		(11,1.6)
		(12,1.6)
	};
	\addplot  coordinates {
		(1,1)
		(2,1.13)
		(3,1.2)
		(4,1.24)
		(5,1.26)
		(6,1.27)
		(7,1.5)
		(8,1.53)
		(9,1.75)
		(10,1.76)
		(11,1.78)
		(12,1.78)
	};

	\legend{{\tiny $BTG_U(500)$}, {\tiny $BTG_U(1000)$}, {\tiny $BTG_U(1500)$}, {\tiny $BTG_U(2000)$},{\tiny $BTG_U(2500)$}}

	\end{axis}
\end{tikzpicture}};

\end{tikzpicture}

\caption{Speedup for ten different Datalog versions of the \texttt{BTG} program, with different parameters. 
\textbf{Left.} Speedup of programs generated with the \texttt{BTA} algorithm. 
\textbf{Right.} Speedup of programs generated with the \texttt{UTA} algorithm. 
The values $X$ of \texttt{BTG}$_B(X)$ and \texttt{BTG}$_U(X)$ indicate the number of clauses in the Datalog program.}\label{fig:speedup1}
\end{figure}

We notice several differences between the results obtained for programs generated with BTA and UTA algorithms.
First of all, the number of clauses 
of the programs is bigger in the UTA cases and the number of leaves is similar in both BTA and UTA cases (cf. Figure~\ref{fig:leaves}); 
however, the runtime are considerable smaller for the UTA programs. 
This is due to the fact that the binary trees described by unbalanced Datalog programs are smaller (there are fewer leaves but more trees), hence, their parallel and-or trees are smaller and the derivation of the trees is faster.

Another difference is the impact of increasing the number of threads (or expand threads): for UTA programs, the maximum speedup is $1.78$, which is not as good as for the BTA programs. 
The reason is again the size of the and-or trees. As the unbalanced trees are smaller, their 
creation is a small computational task and, therefore, the sequential overhead incurred by starting, syncing and distributing work among threads 
cannot be offset by working in parallel. 

Apart from the size, there is another reason which prevents the speedup when increasing the number of expand threads for the UTA programs.
In principle, keeping every expansion thread busy expanding a concrete part of a tree instead of directing it to work in different parts of the tree
results in the best speedup. An ideal implementation would 
adapt to the tree shape and size, and would 
dedicate new expansion threads only for computations of sufficiently large parts of the tree. 
However, in unbalanced trees, it cannot be known in advance if part of a tree 
is large enough to offset the setup costs of dedicating a new thread to it, instead of just executing the work in the current thread. As a result, 
when a new thread is dedicated to expand a part of a tree that is not big enough, the expansion process is slowed down and execution time increases.

\begin{figure}

\begin{tikzpicture}
\draw (0,0) node{

\begin{tikzpicture}[scale=.75]
	\begin{axis}[xlabel=number of leaves,
		ylabel=time,legend style={at={(0.1,.9)},anchor=north west}]
	\addplot  coordinates {
		(3960,8)
		(10160,65)
		(16360,212)
		(22560,444)
		(28760,760)
	};
	\addplot coordinates {
		(3960,7)
		(10160,55)
		(16360,178)
		(22560,366)
		(28760,625)
	};
	\addplot coordinates {
		(3960,5)
		(10160,32)
		(16360,100)
		(22560,204)
		(28760,346)
	};
	\addplot  coordinates {
		(3960,3)
		(10160,17)
		(16360,52)
		(22560,108)
		(28760,184)
	};
        \legend{{\tiny 1t}, {\tiny 6t}, {\tiny 6t+1e}, {\tiny 6t+6e}}

	\end{axis}

\end{tikzpicture}};

\draw (7,0) node{
\begin{tikzpicture}[scale=.75]
\begin{axis}[xlabel=number of leaves,
		ylabel=time,legend style={at={(0.1,.9)},anchor=north west}]
	\addplot  coordinates {
		(6656, 5)
		(13352, 18)
		(20590, 50)
		(28518, 104)
		(35248, 196)
	};
	\addplot  coordinates {
		(6656, 4)
		(13352, 15)
		(20590, 39)
		(28518, 81)
		(35248, 154)
	};
	\addplot  coordinates {
		(6656, 4)
		(13352, 14)
		(20590, 36)
		(28518, 72)
		(35248, 130)
	};
	\addplot  coordinates {
		(6656, 4)
		(13352, 14)
		(20590, 34)
		(28518, 65)
		(35248, 110)
	};
	
	 \legend{{\tiny 1t}, {\tiny 6t}, {\tiny 6t+1e}, {\tiny 6t+6e}}

	\end{axis}
\end{tikzpicture}};

\end{tikzpicture}

\caption{
Effect of parallelism on growing Datalog programs. 
\textbf{Left.} Time (in seconds) for 
the first benchmark (balanced trees). \textbf{Right.} Time (in seconds) for the second benchmark (unbalanced trees).}\label{fig:leaves} 
\end{figure}
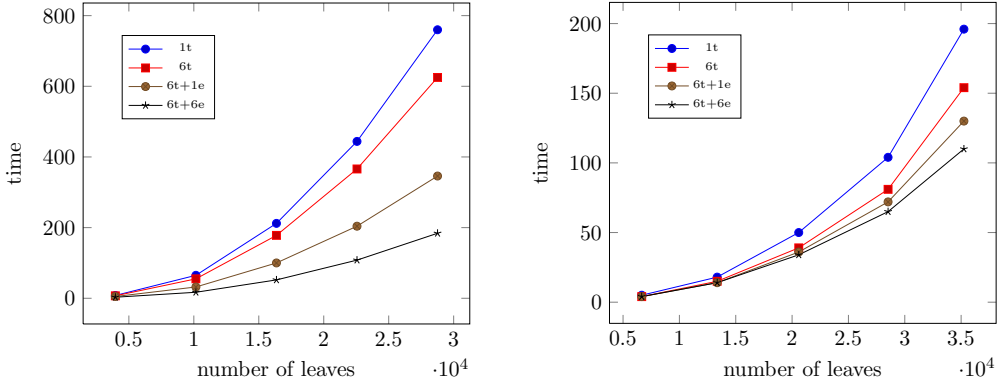

\textbf{3. Effect of data growth} or \emph{would parallelisation bring benefits when the database described by the program increases?}
Suppose one works with dynamically growing data and does not know in advance whether the database will eventually resemble a well-parallelisable 
case like BTA or a badly parallelisable case like UTA. Would CoALP do any good in the pessimistic scenario?  Figure~\ref{fig:speedup1} shows 
that the speedup improves with the growth of the database: UTA programs of size $500$ and $1000$ clauses  hardly allow any speedup, but with 
$2500$ clauses, the speedup is nearing $2$ times. Figure~\ref{fig:leaves} studies this effect under a different angle: assuming increase in
database and hence the number of and-or tree leaves, would parallel execution reduce the growth of execution time? -- and it does for both BTA and UTA experiments.

\section{Fibrational Semantics for Parallelism}\label{sec:cotrees}

We proceed to extend our coalgebraic approach to the general first-order case.
Unification and SLD-resolution algorithms are P-complete in the general case~\cite{DKM84,Kanellakis88}.
In practical terms, P-completeness of an algorithm means that its parallel implementation would not provide effective speedup. The problem  can be illustrated 
using the following example.

\begin{example}\label{ex:unsound}
 The sequential derivation for the goal { \verb"?- btree(tree(X,X,R))"} fails due to ill-typing. But, if the proof search proceeds in parallel fashion, 
 it may find substitutions for {\verb"X"} in distinct parallel and-branches of the derivation tree, see Figure~\ref{fig:treeu}. 
These substitutions will give an unsound result.

\end{example}

Therefore, implementations of \emph{parallel} SLD-derivations require keeping special records of previously made substitutions and, hence, involve additional
data structures and algorithms 
that coordinate variable substitution in different branches of parallel derivation trees \cite{GPACH12}. 
This reduces 
LP capacity for parallelisation.

The practical response to this problem (cf.~\cite{GPACH12}), has been to distinguish cases where parallel SLD-derivations can be sound without synchronisation, 
and achieve efficient parallelisation there. 
There are two kinds of and-parallelism: \emph{independent} and-parallelism and  \emph{dependent} and-parallelism. The former
arises when, given two or more subgoals, there is no common variable in these goals (cf.~Example~\ref{ex:binarytree}). On the contrary, \emph{dependent} and-parallelism 
appears when two or more subgoals have a common variable and compete in the creation of bindings for such a variable (cf. Example~\ref{ex:unsound}). Dependent and-parallelism 
can produce unsound derivations due to variable dependencies; hence,  and-parallel implementations have to sacrifice parallelism for soundness and synchronise dependent variables. 
A different approach is pursued by CoALP~\cite{KP11,KP11-2} below.

\begin{figure}

\begin{minipage}{.45\textwidth}
 
 \begin{center}
  
\begin{tikzpicture}[level 1/.style={sibling distance=15mm},
level 2/.style={sibling distance=30mm},
level 3/.style={sibling distance=5mm},scale=.8,font=\footnotesize]
  
  \node (root) {\texttt{btree(tree(X,X,R))}} [level distance=6mm]
     child { [fill] circle (2pt)
             child { node {\texttt{btree(X)}} 
                     child { [fill] circle (2pt)
                             child { node {$\Box$} }
                           }
                   }
             child { node {\texttt{bit(X)}} 
                     child { [fill] circle (2pt)
                             child { node {$\Box$} }
                           }
                   }
             child { node {\texttt{btree(R)}} 
                    child { [fill] circle (2pt)
                             child { node {$\Box$} }
                           }
                   }
           }
           
   ;
\end{tikzpicture}

 \end{center}

\end{minipage}
\begin{minipage}{.55\textwidth}
\begin{center}

\begin{tikzpicture}[level 1/.style={sibling distance=15mm},
level 2/.style={sibling distance=30mm},
level 3/.style={sibling distance=5mm},scale=.8,font=\footnotesize]
  
  \node (root) {\texttt{btree(tree(X,X,R))}} [level distance=6mm]
     child { [fill] circle (2pt)
             child { node {\texttt{btree(X)}}}
             child { node {\texttt{bit(X)}} }
             child { node {\texttt{btree(R)}}}
           }
           
   ;
\end{tikzpicture}
\end{center}

\end{minipage}

\caption{\footnotesize{\textbf{Left:} The naive first-order parallelisation leading to unsound substitutions, if two branches of the tree are allowed to substitute for 
$X$ independently and in parallel. \textbf{Right}: the action of $\bar{p}$ on the goal \texttt{btree(tree(X,X,R))}; also, a coinductive tree for \texttt{btree(tree(X,X,R))}.}}
\label{fig:treeu} 
\end{figure}
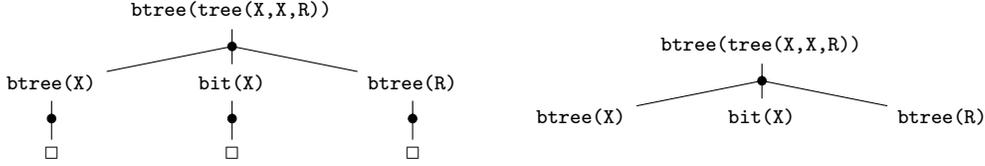

\subsection{Coalgebraic Semantics for First-order Parallel Derivations}\label{sec:fib}

We briefly recall the constructions involved in the coalgebraic semantics for first-order logic programs. We then translate them into derivation algorithms and implementation,
following the same scheme as in the previous section. 

Following the standard practice, we model the
first-order language underlying a logic program by a Lawvere
theory~\cite{AmatoLM09,BonchiM09,BruniMR01}.

\begin{definition}\label{Law}
  Given a signature $\Sigma$ of function symbols, the {\em Lawvere
    theory} $\mathcal{L}_{\Sigma}$ generated by $\Sigma$ is the
  following category: $\texttt{ob}(\mathcal{L}_{\Sigma})$ is the set
  of natural numbers.  For each natural number $n$, let $x_1,\ldots
  ,x_n$ be a specified list of distinct variables. Define
  $\texttt{ob}(\mathcal{L}_{\Sigma})(n,m)$ to be the set of $m$-tuples
  $(t_1,\ldots ,t_m)$ of terms generated by the function symbols in
  $\Sigma$ and variables $x_1,\ldots ,x_n$. Define composition in
  $\mathcal{L}_{\Sigma}$ by substitution.
%
\end{definition}

\begin{example}\label{ex:arrows}
  Consider \texttt{BinaryTree}.  
  The constants \texttt{O}, \texttt{1} and \texttt{empty} are modelled by maps
  from $0$ to $1$ in $\mathcal{L}_{\Sigma}$,  and \texttt{tree} is modelled by a map from
  $3$ to $1$. The term \texttt{tree(0,0,empty)} is therefore modelled by the map
  from $0$ to $1$ given by the composite of the maps modelling
  \texttt{tree}, \texttt{0} and \texttt{empty}.
	\end{example}

For each signature $\Sigma$, we extend the set
$At$ of atoms for a ground logic program to the functor $At:\ls^{op}
\rightarrow Set$ that sends a natural number $n$ to the set of all
atomic formulae generated by $\Sigma$, variables among a fixed set
$x_1,\ldots ,x_n$, and
the predicate symbols appearing in the logic program. A map $f:n
\rightarrow m$ in $\ls$ is sent to the function $At(f):At(m)
\rightarrow At(n)$ that sends an atomic formula $A(x_1, \ldots,x_m)$
to $A(f_1(x_1, \ldots ,x_n)/x_1, \ldots ,f_m(x_1, \ldots ,x_n)/x_m)$,
i.e., $At(f)$ is defined by substitution.

\begin{example}\label{ex:At}
For \texttt{BinaryTree}, $At(3)$ is a poset containing:
 \texttt{bit(0),} \texttt{bit(1),} \texttt{btree(empty),} \texttt{btree(tree(L,X,R)),} \texttt{btree(L),} \texttt{bit(X),} \texttt{btree(tree(0,0,0)),}
 \texttt{btree(tree(1,1,1)),} \texttt{btree(tree(0,empty,0)),} \ldots, \texttt{bit(tree(0,0,0)),} \texttt{bit(tree(1,1,1)),} \texttt{bit(tree(0,empty,0)),}  \ldots, 
 \texttt{btree(tree(tree(L,X,R)),X,tree(L,X,R)),}  \texttt{btree(tree(tree(X,X,X)),X,tree(X,X,X)),} \texttt{...} .
Due to the presence of the function symbol \texttt{tree} that can be composed recursively,  it is an infinite set. Notice the restriction on the number of distinct
variables in the fibre of $3$.
\end{example}

Given a logic program $P$ with function symbols in $\Sigma$, \cite{KP11,KPS12-2} model $P$ by the $Lax(\ls^{op},Poset)$-coalgebra,
 whose $n$-component takes an atomic formula $A(x_1,\ldots ,x_n)$ with at
most $n$ variables, considers all substitutions of clauses in $P$
whose head \emph{agrees} with $A(x_1,\ldots ,x_n)$, and gives the set of sets
of atomic formulae in antecedents. 
We say a head $H$ (from a clause $H \gets body$) agrees with $A(x_1,\ldots ,x_n)$
if the following conditions hold:

\begin{enumerate}
	\item $H\theta = A(x_1,...,x_n)$ 

\item applying $\theta$ to $body$ yields formulae all of whose variables must be among $x_1, \ldots ,x_n$.
\end{enumerate}

The conditions $(i,ii)$ above deserve careful discussion, as they have implications on implementation of CoALP. They are necessary and sufficient conditions to insure that 
the fibrational discipline is obeyed in the coalgebraic model.  
Note that condition $(i)$ resembles very much the definition of term-matching (cf.~Section~\ref{sec:backgr}). The asymmetric application of the substitution $\theta$ is 
not the only feature that distinguishes the above items from most general unifiers (mgus): item $(i)$ does not require $\theta$ to be \lq\lq{}most general\rq\rq{}, which 
also distinguishes $(i)$ from term-matching.
Item $(ii)$ insures no new variables have been introduced within one fibre.

\begin{example}\label{ex:CA}
Continuing Example \ref{ex:At} for the program \texttt{BinaryTree}, for $At(0)$, the $Lax(\ls^{op},Poset)$-coalgebra will essentially give account to 
the version of \texttt{BTG} program in which all possible ground instances of  \texttt{BinaryTree} are present in the form of clauses:
e.g, 
for 
\texttt{btree(tree(empty,0,empty))} $\in At(0)$,  $p(0)(\texttt{btree(tree(empty,0,empty))})= \{\{\texttt{btree(empty), bit(0), btree(empty)}\}\}$.
For 
\texttt{btree(tree(tree(empty,0,empty),0,tree(empty,0,empty)))} $\in At(0)$,\\   
$p(0)(\texttt{btree(tree(tree(empty,0,empty),0,tree(empty,0,empty)))}) = \{\{\texttt{btree(tree(empty,0,empty)), bit(0), btree(tree(empty,0,empty))}\}\}$.

But note that the empty set will correspond to e.g. $p(1)(\texttt{bit(X)})$, as no clause in \texttt{BinaryTree} \emph{agrees} with it; cf. Item $(i)$.

\end{example}

The next example shows the effect of items $(i)-(ii)$ on the derivations. 

\begin{example}\label{ex:TQ}
Consider the following program (we call it TQ):
\begin{verbatim}
 1. T(X,c) :- Q(X).
 2. Q(X) :- P(X).
 3. Q(a).
 4. P(b) :- P(X).
\end{verbatim}
%
%

For $Q(a) \in At(0)$, 
$p(0)(Q(a))= \{\{P(a)\}, \{\} \}$, the last set $\{\}$ models the third clause.
However, for $Q(X) \in At(1)$, the corresponding set will be
$p(1)(Q(X)) = \{\{P(X)\}\}$. Item $(ii)$ plays a role with clauses introducing a new variable in the body, like clause 4. For fiber of $0$, 
$p(0)(P(b)) = \{\{P(a)\},\{P(b)\},\{P(c)\}\}$, note that both the fact that $\theta$ is not required to be an mgu, and the item $(ii)$ play a role here.
See also Figure~\ref{fig:TQ}, the right-hand side.

\end{example}

\begin{figure}
 \begin{center}
 \begin{tikzpicture}[level 1/.style={sibling distance=15mm},
level 2/.style={sibling distance=30mm},
level 3/.style={sibling distance=5mm},scale=.8,font=\footnotesize]
  
  \node (root) {\texttt{T(X,c)}} [level distance=10mm]
             child { node {\texttt{Q(X)}} 
               child { node {\texttt{P(X)}} 
							      child {node {\texttt{P(b)}}
										   child {node {\texttt{P(b)}}
											     child  { node {$\vdots$}}}}}
               child { node {$\Box$} }}  ;
\end{tikzpicture}
 \begin{tikzpicture}[level 1/.style={sibling distance=15mm},
level 2/.style={sibling distance=30mm},
level 3/.style={sibling distance=5mm},scale=.8,font=\footnotesize]
  
  \node (root) {\texttt{T(X,c)}} [level distance=6mm]
	child { [fill] circle (2pt)
             child { node {\texttt{Q(X)}}
						child { [fill] circle (2pt)
               child { node {\texttt{P(X)}} 
							    }}}};
\end{tikzpicture}
 \begin{tikzpicture}[level 1/.style={sibling distance=15mm},
level 2/.style={sibling distance=20mm},
level 3/.style={sibling distance=20mm},scale=.8,font=\footnotesize]
  
  \node (root) {\texttt{T(a,c)}} [level distance=6mm]
	child { [fill] circle (2pt)
             child { node {\texttt{Q(a)}} 
						child { [fill] circle (2pt)
						   child { node {\texttt{P(a)}}} }
							child { [fill] circle (2pt)
               child { node {$\Box$} 
							    }}}};
\end{tikzpicture}
 \begin{tikzpicture}[level 1/.style={sibling distance=15mm},
level 2/.style={sibling distance=20mm},
level 3/.style={sibling distance=20mm},scale=.8,font=\footnotesize]
  
  \node (root) {\texttt{T(b,c)}} [level distance=6mm]
	child { [fill] circle (2pt)
             child { node {\texttt{Q(b)}} 
						child { [fill] circle (2pt)
						   child { node {\texttt{P(b)}}
							child { [fill] circle (2pt)
               child { node {\texttt{P(X)}} 
							    }}}}}};
\end{tikzpicture}
 \begin{tikzpicture}[level 1/.style={sibling distance=15mm},
level 2/.style={sibling distance=20mm},
level 3/.style={sibling distance=20mm},scale=.8,font=\footnotesize]
  
  \node (root) {\texttt{T(b,c)}} [level distance=6mm]
	child { [fill] circle (2pt)
             child { node {\texttt{Q(b)}} 
						child { [fill] circle (2pt)
						   child { node {\texttt{P(b)}}
							child { [fill] circle (2pt)
               child { node {\texttt{P(a)}} 
							    }}
							    	child { [fill] circle (2pt)
               child { node {\texttt{P(b)}} 
							    		child { [fill] circle (2pt)
               child { node {\texttt{P(a)}} 
							    }}
							    	child { [fill] circle (2pt)
               child { node {\texttt{P(b)}}     child { node {$\vdots$} }
							    							    }}
							    	child { [fill] circle (2pt)
               child { node {\texttt{P(c)}} 
							    }}
							    }}
							    	child { [fill] circle (2pt)
               child { node {\texttt{P(c)}} 
							    }}
							    							    }}}};
\end{tikzpicture}
 \end{center}
\caption{\footnotesize{\textbf{Far Left}. The SLD-tree for TQ and the goal \texttt{T(X,c)}. \textbf{Centre.} The coinductive trees for TQ and the goals \texttt{T(X,c)},
\texttt{T(a,c)},  \texttt{T(b,c)}. For \texttt{T(X,c)} and \texttt{T(a,c)} the coinductive trees correspond to the  action of $\bar{p}(1)(T(X,c))$ and $\bar{p}(0)(T(a,c))$. 
\textbf{Far Right.} Action of $\bar{p}(0)(\texttt{T(b,c)})$.  }} 
\label{fig:TQ} 
\end{figure}
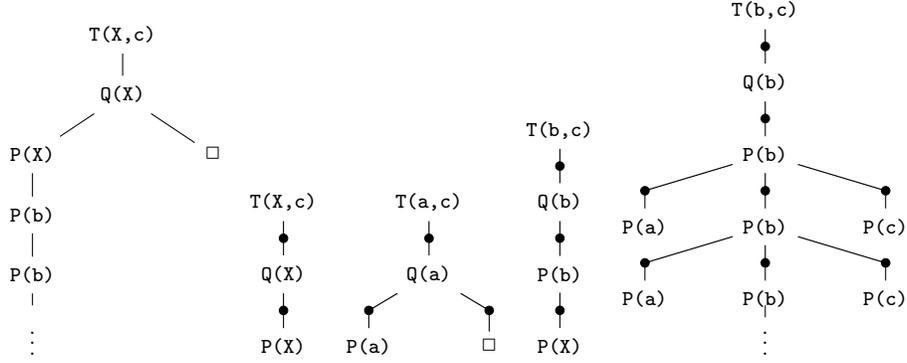


Similarly to the ground case, this can be extended to model derivations. A lax natural transformation ${\bar p}:At
\longrightarrow C(P_cP_f)At$, when evaluated at $n$, is the function
from the set $At(n)$ to the set of all possible derivations in a fibre $n$. Note the $P_c$ -- a countable powerset functor -- is introduced to model countable signature, 
as will be further explained below.
The role of laxness is explained in detail in \cite{KP11,KPS12-2} and critiqued in \cite{BonchiZ13}, so we will not focus on this issue here.
Instead, we will concentrate on the role of the \emph{fibrational semantics} 
in the design of parallel derivation algorithm. 

\begin{example}\label{ex:treefo}
  Consider \texttt{BinaryTree} as in Example~\ref{ex:binarytree}. Suppose we
  start with an atomic formula
  \verb"btree(tree(X,X,R))" $\in At(2)$. Then ${\bar p}($\verb"btree(tree(X,X,R))"$)$ is the
  element of $C(P_cP_f)At(2)$ expressible by
    the tree on the right hand side of Figure~\ref{fig:treeu}.
Note that despite there being an infinite number of binary trees and the infinite number of instances of clauses of \texttt{Binary Tree}, the countability accounted
for by $P_c$ does not arise for programs of this kind, cf. Items $(i-ii)$.  

    This tree agrees partially with the and-or parallel
    derivation tree for $\mathtt{btree(tree(X,X,R))}$ given on the left of Figure~\ref{fig:treeu}. 
    But it has
    leaves \texttt{btree(X)}, \texttt{bit(X)} and \texttt{btree(R)} (cf. Item $(i)$),
    whereas the and-or parallel derivation tree follows those nodes,
    using substitutions determined by mgu's that might not be
    consistent with each other, e.g., there is no consistent
    substitution for \texttt{X}.

Action of  $\bar{p}(1)(\texttt{btree(tree(empty,empty,R))})$ and $\bar{p}(1)(\texttt{btree(tree(0,0,empty))})$,  are shown in Figure~\ref{fig:fiber}. 
Compare with Examples~\ref{ex:At} and~\ref{ex:CA}; again, note the effect of Item $(i)$.

\end{example} 

\begin{example}
The action of $\bar{p}(0)T(b,c)$ is given in Figure~\ref{fig:TQ}. Note the infinite depth. If there was a constructor in the language, e.g. \texttt{f},
then the branching could be infinite, bringing $\texttt{P(f(a))}, \texttt{P(f(f(a))}, \ldots , \texttt{P(f(b))}, \ldots $ in the fibre of $0$, alongside 
the branches with \texttt{P(a)}, \texttt{P(b)}, and \texttt{P(c)}. This is the effect of Item $(i)$ not requiring $\theta$ to be an mgu. Programs of this
kind require countability in $P_cP_f$.
\end{example}

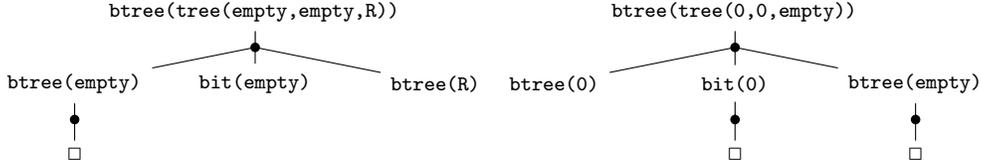
\begin{figure}
 \begin{center}
 \begin{tikzpicture}[level 1/.style={sibling distance=15mm},
level 2/.style={sibling distance=30mm},
level 3/.style={sibling distance=5mm},scale=.8,font=\footnotesize]
  
  \node (root) {\texttt{btree(tree(empty,empty,R))}} [level distance=6mm]
     child { [fill] circle (2pt)
             child { node {\texttt{btree(empty)}} 
                     child { [fill] circle (2pt)
                             child { node {$\Box$} }
                           }
                   }
             child { node {\texttt{bit(empty)}}                   
                   }
             child { node {\texttt{btree(R)}} 
                                      }
           }  ;
\end{tikzpicture}
\begin{tikzpicture}[level 1/.style={sibling distance=15mm},
level 2/.style={sibling distance=30mm},
level 3/.style={sibling distance=5mm},scale=.8,font=\footnotesize]
  \node (root) {\texttt{btree(tree(0,0,empty))}} [level distance=6mm]
     child { [fill] circle (2pt)
             child { node {\texttt{btree(0)}} 
                   }
             child { node {\texttt{bit(0)}} 
                     child { [fill] circle (2pt)
                             child { node {$\Box$} }
                           }
                   }
             child { node {\texttt{btree(empty)}} 
                    child { [fill] circle (2pt)
                             child { node {$\Box$} }
                           }
                   }
           } ;
\end{tikzpicture}
 \end{center}
\caption{\footnotesize{Action of  $\bar{p}$(1) on \texttt{btree(tree(empty,empty,R))} and \texttt{btree(tree(0,0,empty))}; also, 
coinductive trees for \texttt{btree(tree(empty,empty,R))} and \texttt{btree(tree(0,0,empty))}.}}
\label{fig:fiber} 
\end{figure}

We define a clause to be \emph{regular}, if the set of variables appearing in its body is a proper subset of the variables appearing in its head. 
We define logic program to be a \emph{regular program}, if all of its clauses are regular; otherwise the program is \emph{irregular}.

\subsection{From Semantics to Derivation Algorithm}

The  above coalgebraic semantics suggests to restrict  \emph{unification} involved in e.g. SLD-derivations 
to \emph{term matching}. 
However, for constructive reasons, we do not follow the above coalgebraic semantics literally this time,
and do not lift restriction of the term-matcher to be the mgu. This permits to avoid infinite cycles when working with irregular programs, cf. Figure~\ref{fig:TQ}.
To compensate, we relax the restriction on body variables (cf. Item $(ii)$) and allow new variables to be introduced within the derivations.

\begin{definition}\label{df:coindt}
Let $P$ be a logic program and $G=\gets A$ be an atomic goal.
The \emph{coinductive tree} for $A$ is
  a possibly infinite tree $T$ satisfying the following properties.
\begin{itemize}
\item $A$ is the root of $T$.
\item Each node in $T$ is either an and-node (an atom) or an or-node (given by $\bullet$).
\item For every and-node $A'$ occurring in $T$, if there exist exactly $m>0$ 
distinct  clauses $C_1, \ldots , C_m$ in $P$ (a clause $C_i$ has the form $B_i
  \gets B^i_1, \ldots , B^i_{n_i}$, for some $n_i$),  such that $A' = B_1\theta_1 =
  ... = B_m\theta_m$, for mgus $\theta_1, \ldots , \theta_m$,  then $A'$ has exactly $m$ children given by
  or-nodes, such that, for every $i \in m$, the $i$th or-node has $n_i$
  children given by and-nodes $B^i_1\theta_i, \ldots ,B^i_{n_i}\theta_i$.
\end{itemize}
\end{definition}

\begin{example}\label{ex:CT}
Figures~\ref{fig:treeu}--\ref{fig:fiber} (excluding Figure~\ref{fig:TQ}) showing the action of $\bar{p}$ on program atoms, would correspond exactly to coinductive 
trees for these atoms. Note 
the difference between the coinductive trees and and-or parallel trees (Figure~\ref{fig:treeu}) and the SLD-derivations for \texttt{BinaryTree} (Figure~\ref{fig:sld-derivation}).
However, there will be cases 
when the coinductive trees and the corresponding action of $\bar{p}$ differ, cf. Figure~\ref{fig:TQ}. 
\end{example}


The seemingly small restriction of unification to term-matching changes the way the proof-search is handled within each coinductive tree. Note that unification in general is inherently
sequential, whereas term matching is parallelisable~\cite{DKM84}.
Term-matching permits implicit handling of both parallelism and corecursion:

\begin{itemize}
	\item due to term-matching, all existing variables are not instantiated and will remain the same within one tree, and therefore no explicit variable synchronisation
	is needed when (the branches of) trees are expanded 
	in parallel.
	\item term-matching permits to unfold coinductive trees lazily, keeping
each individual tree at a finite size, provided the program is
well-founded~\cite{KP11-2,KPS12-2}. Laziness in its turn plays a role in delaying substitutions in parallel derivations. 
\end{itemize}
 
\begin{example}\label{ex:TQs}
Consider the program TQ from Example~\ref{ex:TQ}. Figure~\ref{fig:TQ} shows the SLD-derivations 
and the coinductive tree for goal $T(X,c)$. The SLD-derivations produce one derivation that loops forever; and, subject to clause re-ordering or loop-termination, would also give
one answer computed on the second step of derivations.
The coinductive tree for the same goal stops lazily, and neither loops nor gives the answer.    
However, after a suitable substitution, the coinductive tree for \texttt{T(a,c)} gives the sought answer (we will use this technique for coinductive derivations in the next section).
The coinductive trees  for the goal atoms \texttt{T(X,c)} and \texttt{T(a,c)} in Figure~\ref{fig:TQ} correspond to the action of $\bar{p}$ on them. However, the case of \texttt{T(b,c)} 
is different.
Note the loop stopping lazily in the coinductive tree for \texttt{T(b,c)}, which distinguishes it from infinite SLD-derivations and the construction of $\bar{p}(0)(\texttt{T(b,c)})$. 

\end{example}  


In line with Section~\ref{sec:par}, we would like to establish a \emph{constructive completeness result}, showing how to transform our coalgebraic semantics into coinductive trees.
However, this is impossible;  Examples~\ref{ex:TQ} and \ref{ex:TQs}, and Figure~\ref{fig:TQ} show a counter-example.  In fact, in \cite{KP11-2,KPS12-2}, 
completeness theorem for CoALP gives a weak, rather than a strong (constructive), completeness statement.
We can, however, establish 
constructive completeness for regular LPs.

\begin{lemma}[Restricted Constructive Completeness]\label{lem:CC}
Let $P$ be a regular logic program, and $G$ be an atomic goal with exactly $s$ variables. Given the construction of $\overline{p}(s)(G)$,
there exists (can be constructed) a coinductive  tree $T_G$ for $G$, such that:
\begin{itemize}
\item (Tree depth 0.) The root of $T_G$ is given by $p_0(s)(G) = G$.
\item (Tree depth $n$, for odd $n$.) Every node $A$ appearing at the tree depth $n-1$ has $m$ $\bullet$-child-nodes at the tree depth $n$, corresponding to the number of sets 
contained in the set $p(s)(A)$.
\item (Tree depth $n$, for even $n>0$.) Every $i$th $\bullet$-node at the depth $n-1$ with a parent node $A$ at the depth $n-2$ has children at the depth $n$, given by
the distinct elements of the $i$th set in $p(s)(A)$.
\end{itemize}
Moreover,  $T_G$  has finite depth $2n$ (for some $n \in \mathcal{N}$)
   iff ${\bar p}(G) = p_n(s)(G)$. 
 The $T_G$ is infinite iff  ${\bar p}(s)(G)$ is given
   by the element of the limit $\lim_{\omega}(p_n)(s)(At)$ of an infinite
   chain given by Construction of $C(P_fP_f)$.   
\end{lemma}
\begin{proof}
The core of this inductive proof is an observation that, if $\theta$ is restricted to the variables of $A$ as it effectively happens for regular programs, and $A\theta = B$, 
then necessarily $\theta$ is an mgu. This fact eliminates the difference between Item $(i)$ and the term-matching for coinductive trees. Further, Item $(ii)$ is redundant 
for regular programs. Note that we use $P_fP_f$, rather than $P_cP_f$ here.
\end{proof}

\subsection{From Derivation Algorithm to Implementation}

The CoALP implementation of coinductive trees reuses the notion of clause-trees (cf. Definition~\ref{def:clause-tree})
presented in Subsection~\ref{subsec:and-or-g}. 

\begin{example}
Figure~\ref{pic:clausetree} shows the clause-trees for all the clauses of the \texttt{BinaryTree} program.
 
 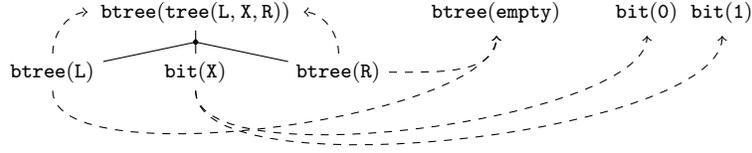
\begin{figure}
\begin{center}
\footnotesize{
  \begin{tikzpicture}[scale=0.2,font=\footnotesize,baseline=(current bounding box.north),grow=down,level distance=20mm]
      \node (btree) {$\mathtt{btree(tree(L,X,R))}$}
        child {[fill,sibling distance=95mm] circle (4pt)
            child { node (btreeL) {$\mathtt{btree(L)}$} }
            child { node (btreeX) {$\mathtt{bit(X)}$} }
            child { node (btreeR) {$\mathtt{btree(R)}$} } };
      \node [right of=btree,node distance=4cm] (btreeE){$\mathtt{btree(empty)}$};
      \node [right of=btreeE,node distance=2cm] (bit0) {$\mathtt{bit(0)}$};
      \node [right of=bit0,node distance=1cm] (bit1) {$\mathtt{bit(1)}$};
      \draw[->,dashed,shorten >=2pt]
          (btreeX) to[out=270,in=270,looseness=0.5] (bit0);
      \draw[->,dashed,shorten >=2pt]
          (btreeX) to[out=270,in=270,looseness=0.5] (bit1);
      \draw[->,dashed,shorten >=2pt]
          (btreeL) to[out=90,in=180] (btree);
      \draw[->,dashed,shorten >=2pt]
          (btreeL) to[out=270,in=270,looseness=0.5] (btreeE);
      \draw[->,dashed,shorten >=2pt]
          (btreeR) to[out=90,in=0] (btree);
      \draw[->,dashed,shorten >=2pt]
          (btreeR) to[out=0,in=270] (btreeE);
\end{tikzpicture}
}
\end{center}
\vspace{-1.5\baselineskip}
\caption{\footnotesize{Clause-trees for the \texttt{BinaryTree} program with dashed lines denoting open list references.}}
\label{pic:clausetree} 
\end{figure}
 
\end{example}

Clause-trees of Definition~\ref{def:clause-tree} remain to be the building blocks of coinductive trees. However, the generation of 
coinductive trees from clause-trees is slightly different to the one presented in Construction~\ref{cons:and-or}. We assume all programs are pre-processed this way.

\begin{construction}[Go-coinductive tree]\label{cons:cotree}
Given a logic program $P$ and a goal $G=\gets A$, generate a Go-coinductive tree $T$ as follows:

\begin{enumerate}
 \item A root $A$ for $T$ is created as an and-node containing the goal atom.
  \item The open list of the root $A$ is constructed by adding references to all clause-trees that have a unifiable root atom.
 \item For each reference in an open list $O$ of a node $A'$ where the corresponding atom matches the referenced root node's atom $R$, a copy of the or-node below the
 referenced node and all its children in the clause-tree are added as child to $A'$.  All the substitutions that were needed to make $R$ match are also applied to the newly 
 copied node atoms. The reference is then deleted from $O$.
 \item This process continues until all references in all the open lists in the tree $T$ have been processed.
\end{enumerate}
\end{construction}

The first difference to the construction of and-or parallel trees is that instead of needing equality to the root node of a referenced tree its now only required 
to be term matching in order to expand the tree. The second difference is that the newly added nodes to the tree require application of the substitutions needed 
for term-matching of the clause-tree they belong to. Since requiring equality in the Construction~\ref{cons:and-or} did not generate substitutions, this step was
previously unnecessary.

\begin{example}
 
Given the query \texttt{btree(tree(X,X,R))} in the \verb"BinaryTree" program, we construct its Go-coinductive tree as follows. We start with a tree that consists 
of the goal atom as root and-node and no or-node children. Then, we add references to all clause-trees with unifiable root node atom --- in this case is the clause-tree 
for \texttt{btree(tree(L,X,R))}. For this root to match the query, the substitution \texttt{L/X} is needed. Therefore, we copy the or-node and its and-node children from 
the clause-tree for \texttt{btree(tree(L,X,R))} below our tree root and apply the aforementioned substitution to the copied clause-tree nodes. Now, we should process all
nodes in this newly created tree that have references to other clause-tree roots, but there are no nodes with references that match. Therefore, the process is finished,
and the resulting tree will look like the one depicted in the right side of Figure~\ref{fig:treeu}.

\end{example}

\begin{lemma}\label{lemma:cotreecons}
Let $P$ be a logic program and let $G$ be an atomic goal. Then, the coinductive tree of $G$ is given by Construction~\ref{cons:cotree}.
\end{lemma}

We employ an optimisation technique to minimise the work for constructing Go-coinductive trees. 
When checking a referenced root node for term-matching, it is also checked whether it can be unified. If it is not unifiable, then  term-matching is impossible even if 
substitutions will be applied later when processing the tree. Therefore, non-unifiable references to clause-tree roots can be immediately removed from the open lists.
This process of filtering open lists and copying or-nodes and their child nodes during tree expansion can be done in parallel as no variable substitutions in existing 
atoms is required. It only needs to be guaranteed that only one thread changes the properties of one node concurrently. In the same way, no run-time coordination for
tree merging or variable substitution is needed. However, the parallel expansion option should only be employed on large trees to amortise the additional overhead of dispatching and 
managing multiple threads that execute the mentioned tasks. 

\subsection{Case Study: Ttree}

The construction of coinductive trees for queries in the \verb"BinaryTree" program is too fast to notice any benefit from parallelism; instead we introduce a new example. 

\begin{example}[\texttt{Ttree}]\label{ex:ttree}
This program has three and-parallel branches and thereby can be used to construct coinductive trees with three branches at each non leaf or-node. 

\begin{verbatim}
 1. ttree(0). 
 2. ttree(s(X)) :- ttree(X), ttree(X), ttree(X).
\end{verbatim}

\end{example}

\begin{figure}
\centering
\begin{tikzpicture}[scale=.75]
	\begin{axis}[xlabel= threads (t) and expand threads (e), xlabel style={at={(0.5,-.2)},anchor=north},
		ylabel=speedup,legend style={at={(0.1,.9)},anchor=north west},
		xticklabels={1t,2t,3t,4t,5t,6t,6t+1e,6t+2e,6t+3e,6t+4e,6t+5e,6t+6e},xtick={1,...,12}, x tick label style={rotate=90,anchor=east}]
	\addplot  coordinates {
		(1,1)
		(2,1.12)
		(3,1.19)
		(4,1.19)
		(5,1.23)
		(6,1.23)
		(7,1.68)
		(8,1.94)
		(9,1.94)
		(10,2.05)
		(11,2.17)
		(12,2.17)
	};
	\addplot coordinates {
		(1,1)
		(2,1.10)
		(3,1.17)
		(4,1.18)
		(5,1.2)
		(6,1.21)
		(7,2.25)
		(8,2.78)
		(9,3)
		(10,3.34)
		(11,3.54)
		(12,3.65)
	};
	\addplot coordinates {
		(1,1)
		(2,1.10)
		(3,1.17)
		(4,1.20)
		(5,1.22)
		(6,1.23)
		(7,2.36)
		(8,3.31)
		(9,3.67)
		(10,4.07)
		(11,4.41)
		(12,4.52)
				
	};
	\addplot  coordinates {
		(1,1)
		(2,1.15)
		(3,1.23)
		(4,1.25)
		(5,1.29)
		(6,1.31)
		(7,2.39)
		(8,3.06)
		(9,3.24)
		(10,3.71)
		(11,3.99)
		(12,4.07)
	};

	\legend{{\tiny \texttt{ttree(s$^{11}$(0))} }, {\tiny \texttt{ttree(s$^{12}$(0))} }, {\tiny \texttt{ttree(s$^{13}$(0))} }, {\tiny \texttt{ttree(s$^{14}$(0))} }}	

	\end{axis}
\end{tikzpicture}
\caption{Speedup of the construction of the coinductive tree of \texttt{ttree(s$^i$(0))}, relative to the base case with 1 thread, for different values of $i$ and
different parameters.}\label{fig:speedupttree}
\end{figure}
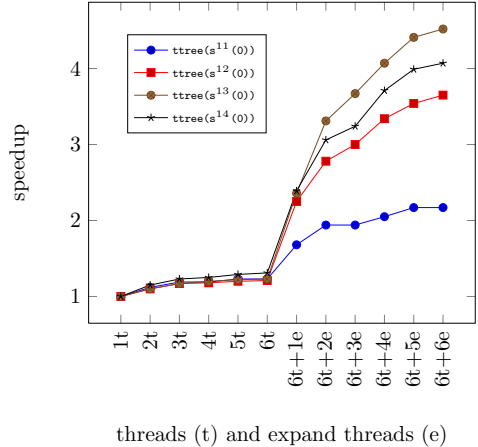

We benchmark our implementation by constructing  the tree and thereby proof for the query \verb"ttree(s"$^i$\verb"(0))", where $i$ indicates the number 
of times that the function \verb"s" is nested (e.g.  \verb"ttree(s"$^2$\verb"(0))" is equal to  \verb"ttree(s(s(0)))"). Given the query \verb"ttree(s"$^i$\verb"(0))",
the associated coinductive tree will have $3^i$ leaf nodes. Therefore, we can expect that the construction of coinductive trees will get advantage of 
a parallel expansion; this is confirmed in Figure~\ref{fig:speedupttree}.

The construction of coinductive trees clearly speeds up with introduction of parallelism. 
Increasing the number of threads speeds up the execution time whether the parallel expand option is activated
(maximum speedup of $4.52$), or not (maximum speedup of $1.31$).
Note that generalising from ground logic programs to first-order logic programs did not reduce the best-case scenario speedup of $\approx 4.5$.

\section{Derivations by Coinductive Trees}\label{sec:coderivations}

In the previous sections, we have seen how first-order atoms, clauses, as well as individual derivations can be modeled using category theoretic constructs.
We have taken inspiration from the coalgebraic semantics to introduce the notion of coinductive trees. 
However, as can be seen from e.g. Figures \ref{fig:TQ} and \ref{fig:fiber}, one coinductive tree may not 
produce the answer corresponding to a refutation by the SLD-resolution. Instead, a sequence of coinductive 
trees may be needed to advance the derivation. 
In this section, we introduce the derivations involving coinductive trees, and discuss their relation to the coalgebraic semantics.
The coalgebraic derivation algorithm follows closely \cite{KP11-2,KPS12-2}, but the implementation we present here is new.

\subsection{Coinductive Derivations: relating semantics to derivations}\label{sec:coder}

We start with  composing coinductive  trees into derivations.

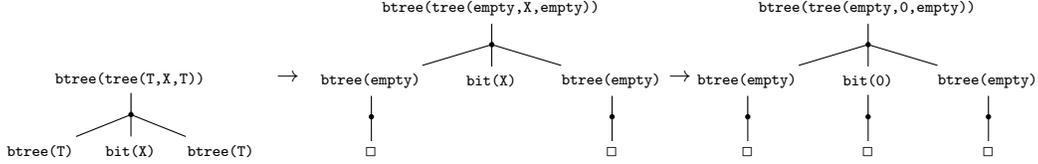
\begin{figure}
\begin{center}
\footnotesize{
  \begin{tikzpicture}[level 1/.style={sibling distance=15mm},
level 2/.style={sibling distance=15mm},
level 3/.style={sibling distance=5mm},scale=.8,font=\tiny]
  \node (root) {\texttt{btree(tree(T,X,T))}} [level distance=6mm]
     child { [fill] circle (1pt)
             child { node {\texttt{btree(T)}}}
             child { node {\texttt{bit(X)}} }
             child { node {\texttt{btree(T)}}}
           };
  \end{tikzpicture}
    \begin{tikzpicture}
     \draw (0,1) node{$\rightarrow$};
     \draw (0,0) node{\textcolor{white}{$\rightarrow$}};
    \end{tikzpicture}
  \begin{tikzpicture}[level 1/.style={sibling distance=15mm},
level 2/.style={sibling distance=20mm},
level 3/.style={sibling distance=5mm},scale=.8,font=\tiny]
  \node (root) {\texttt{btree(tree(empty,X,empty))}} [level distance=6mm]
     child { [fill] circle (1pt)
             child { node {\texttt{btree(empty)}}
                     child {[fill] circle (1pt) 
			    child {node {$\Box$}}
			    }  
                   }
             child { node {\texttt{bit(X)}} }
             child { node {\texttt{btree(empty)}}
						child {[fill] circle (1pt) 
			    child {node {$\Box$}}}
           }};
  \end{tikzpicture}
        \hspace{-.3cm}
    \begin{tikzpicture}
     \draw (0,1) node{$\rightarrow$};
     \draw (0,0) node{\textcolor{white}{$\rightarrow$}};
    \end{tikzpicture}
      \hspace{-.3cm}
  \begin{tikzpicture}[level 1/.style={sibling distance=15mm},
level 2/.style={sibling distance=20mm},
level 3/.style={sibling distance=5mm},scale=.8,font=\tiny]
  \node (root) {\texttt{btree(tree(empty,0,empty))}} [level distance=6mm]
     child { [fill] circle (1pt)
             child { node {\texttt{btree(empty)}}
                     child {[fill] circle (1pt) 
			    child {node {$\Box$}}
			    }  
                   }
             child { node {\texttt{bit(0)}} 
						 child {[fill] circle (1pt) 
			    child {node {$\Box$}}
			    }  }
             child { node {\texttt{btree(empty)}}
                     child {[fill] circle (1pt) 
			    child {node {$\Box$}}
			    }  
                   }
           };
  \end{tikzpicture}
	}
\end{center}
\caption{\footnotesize{A coinductive derivation for the goal \texttt{btree(tree(T,X,T))} and the program \texttt{BinaryTree}. Note that one tree cannot find the correct answer
(substitution), but it needs two steps to compute one possible answer. The first transition is given by $\theta_0 = T/\texttt{empty}$, the second by $\theta_1 = X/\texttt{0}$. }}
\label{pic:tree2} 
\end{figure}

\begin{definition}\label{df:coind-res2}
  Let $G$ be a goal given by an atom $\gets A$ and the 
  coinductive tree $T$ induced by $A$, and let $C$ be a clause $H \gets
  B_1, \ldots , B_n$.  Then goal $G'$ is \emph{coinductively derived}
  from $G$ and $C$ using mgu $\theta$ if the following conditions
  hold:
  \begin{itemize}
   
\item[$\bullet$] $A'$ is an atom in $T$.
\item[$\bullet$] $\theta$ is an \emph{mgu} of $A'$ and $H$.
\item[$\bullet$] $G'$ is given by the atom $\gets A\theta$ and the
  coinductive tree $T'$ determined by $A\theta$.
  
  \end{itemize}

\end{definition}

Coinductive derivations resemble \emph{tree rewriting}, an example is shown in Figure \ref{pic:tree2}. They produce the ``lazy'' corecursive effect:  
derivations are given by potentially infinite number of steps, where each individual step (coinductive tree) is executed in finite time. The ultimate goal of derivations 
is to find success (sub)trees.

\begin{definition}\label{df:derivsub2}
  Let $P$ be a logic program, $G$ be an atomic goal, and $T$ be a
  coinductive tree  determined by $P$ and $G$.  A subtree
  $T'$ of $T$ is called a \emph{coinductive subtree} of $T$ if it
  satisfies the following conditions:
\begin{itemize}
\item the root of $T'$ is the root of $T$ (up to variable renaming);
\item if an and-node belongs to $T'$, then one of its children belongs
  to $T'$.
\item if an or-node belongs to $T'$, then all its children belong to $T'$.
\end{itemize}
A finite coinductive (sub)tree is called a {\em success (sub)tree} if its leaves are empty goals (equivalently, they are followed only by $\Box$ in the usual pictures). 
\end{definition}

Definition of coinductive derivations and refutations is an adaptation of Definition~\ref{df:SLD2}, modulo Definitions~\ref{df:coind-res2} and \ref{df:derivsub2}.

  \begin{example} Figure~\ref{pic:tree2} shows an example of a coinductive derivation. The last coinductive tree is also a success subtree in itself; which means the derivation
  has been successful.
\end{example}
	
Transitions between coinductive trees can be done in a sequential or parallel manner, as we discuss in the next section. In \cite{KPS12-2}, we have proven the soundness and completeness 
of the coalgebraic derivations relative to the coalgebraic semantics of \cite{KMP10,KP11}. 
What would constructive completeness result say here? It would have to build upon Lemma~\ref{lem:CC} when it comes to modelling individual trees in the derivations. In addition, 
we would need to 
produce a construction (algorithm) that, for any goal $G(x_1, \ldots x_k)$ and any map $\theta$ in  
$\mathcal{L}_{\Sigma}^{op}$ (corresponding to a substitution) would produce a coinductive derivation starting at $G(x_1,\ldots , x_k)$ and ending at  $G(x_1,\ldots , x_k)\theta$,
$\theta$ being a  computed substitution. However, this result would be analogous to proving decidability of entailment for the Horn-clause LP, whereas it is only semi-decidable.
This is why, completeness is not generally stated in a constructive form for LP and instead involves the nonconstructive existence assertion, see \cite{Llo88} for 
standard completeness statements and \cite{KPS12-2} for CoALP.

\subsection{From Derivation Algorithm to Implementation}

The construction of coinductive derivations is modelled as a search through the graph of coinductive trees connected by the derivation operation. To keep track of which trees have
to be processed, the implementation maintains an ordered list of coinductive trees that is called the \emph{work queue}. Initially, this list will be filled with the coinductive 
tree constructed from the input goal. The top level control-flow-loop dispatches coinductive trees from the work queue to be processed simultaneously by multiple worker threads.
Worker threads are implemented by goroutines that are lightweight threads of execution and communicate via Go\rq{}s channels for passing values. Each worker executes the
following simplified steps independently on the received tree:

\begin{itemize}
    \item checks and reports if the tree contains a success subtree,
    \item finds nodes with a non-empty open list and compute the set of distinct \emph{mgu's} needed to unify with the referenced clause-tree roots,
    \item for each \emph{mgu} found, applies the \emph{mgu} to a copy of the tree then expands it by the process outlined in construction \ref{cons:cotree},
    \item sends the created trees back to the main control loop to be added back to the work pool for further processing.
\end{itemize}

This process can be run with a variable number of worker threads until either the requested amount of success subtrees are found or the work queue is empty and all the possible
coinductive trees for the query have been processed. This search process requires to take some decisions about strategies to manage the work queue, find open nodes, compacting
coinductive trees and ordering solutions. 



\textbf{Organisation of the work queue.}
The work queue is sufficient to keep track of all trees which have derivation steps that still need to be evaluated.
This list of coinductive trees can be managed either as a first-in-first-out queue or alternatively as a last-in-first-out stack.
This determines the search strategy that is employed to find success trees in the possible derivation chains: 
depth-first in the case of last-in-first-out stack (this is the strategy followed in PROLOG), and breadth-first 
in the case of first-in-first-out queue (which is the strategy followed in CoALP).

A depth-first search strategy has the drawback of not finding some possible success trees if the search forever follows an unrelated infinite derivation chain. By construction, 
coinductive derivation trees are always finite; however, this does not restrict the possibility of an infinite chain of derivations between trees. The breadth-first search is 
the strategy currently chosen in CoALP and therefore no traditional backtracking is employed in contrast to PROLOG. This has the usual drawback of requiring more memory than a 
depth-first search. However, it allows CoALP to easily report solutions sorted by the number of substitutions. If a solution with a specific amount of substitutions is reported, 
it can be guaranteed that no solution with a lesser amount of substitutions will be found later.

\textbf{Strategies to find open nodes.}
Open nodes are and-nodes which contain at least one reference in their open list which points to a root node of a clause-tree with a unifiable atom. Not only leave nodes in the
tree need to be checked but any node. An example program which can give rise to such a situation can be seen in Example~\ref{ex:TQ} where the node for $\mathtt{Q(X)}$ in the tree
for the goal $\mathtt{T(X,c)}$ contains a non empty open list with a reference to $\mathtt{Q(1)}$. There are different strategies to find open nodes. If the left- or right-most
deepest branch is always searched for open nodes, this would mimic PROLOG\rq{}s depth-first search approach to SLD-resolution. To obtain solutions not found in the depth-first
manner, a breadth-first search approach to find open nodes is employed. It always chooses the left most node on the lowest level possible. This will yield a transition to each
possible reachable coinductive tree after a finite number of derivations. This is especially important when dealing with co-inductive or 
cyclic data structures. If no open node in the coinductive tree can be found, the tree will simply be discarded as no further derivations are possible.

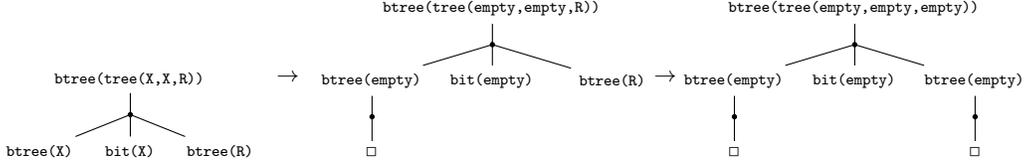
\begin{figure}
\begin{center}
\footnotesize{
  \begin{tikzpicture}[level 1/.style={sibling distance=15mm},
level 2/.style={sibling distance=15mm},
level 3/.style={sibling distance=5mm},scale=.8,font=\tiny]
  \node (root) {\texttt{btree(tree(X,X,R))}} [level distance=6mm]
     child { [fill] circle (1pt)
             child { node {\texttt{btree(X)}}}
             child { node {\texttt{bit(X)}} }
             child { node {\texttt{btree(R)}}}
           };
  \end{tikzpicture}
    \begin{tikzpicture}
     \draw (0,1) node{$\rightarrow$};
     \draw (0,0) node{\textcolor{white}{$\rightarrow$}};
    \end{tikzpicture}
  \begin{tikzpicture}[level 1/.style={sibling distance=15mm},
level 2/.style={sibling distance=20mm},
level 3/.style={sibling distance=5mm},scale=.8,font=\tiny]
  \node (root) {\texttt{btree(tree(empty,empty,R))}} [level distance=6mm]
     child { [fill] circle (1pt)
             child { node {\texttt{btree(empty)}}
                     child {[fill] circle (1pt) 
			    child {node {$\Box$}}
			    }  
                   }
             child { node {\texttt{bit(empty)}} }
             child { node {\texttt{btree(R)}}}
           };
  \end{tikzpicture}
        \hspace{-.3cm}
    \begin{tikzpicture}
     \draw (0,1) node{$\rightarrow$};
     \draw (0,0) node{\textcolor{white}{$\rightarrow$}};
    \end{tikzpicture}
      \hspace{-.3cm}
  \begin{tikzpicture}[level 1/.style={sibling distance=15mm},
level 2/.style={sibling distance=20mm},
level 3/.style={sibling distance=5mm},scale=.8,font=\tiny]
  \node (root) {\texttt{btree(tree(empty,empty,empty))}} [level distance=6mm]
     child { [fill] circle (1pt)
             child { node {\texttt{btree(empty)}}
                     child {[fill] circle (1pt) 
			    child {node {$\Box$}}
			    }  
                   }
             child { node {\texttt{bit(empty)}} }
             child { node {\texttt{btree(empty)}}
                     child {[fill] circle (1pt) 
			    child {node {$\Box$}}
			    }  
                   }
           };
  \end{tikzpicture}
	}
\end{center}
\caption{\footnotesize{A coinductive derivation for the goal \texttt{btree(tree(X,X,R))}. Parallel expansion of nodes in a coinductive tree does not lead to unsound substitutions.}}
\label{pic:tree1} 
\end{figure}

\begin{example}
In Figure \ref{pic:tree1}, the coinductive derivation tree for the goal
\texttt{btree(tree(X,X,R))} is shown on the left. The open leaf on the lowest level to the left is 
\texttt{btree(X)} with the tree templates with the goals \texttt{tree(empty)} and $\mathtt{tree(L_1,B_1,R_1)}$. 
Therefore the distinct \emph{mgu's} are \texttt{$\{$X/empty$\}$} and $\mathtt{\{X/tree(L_1,B_1,R_1)\}}$. They yield different 
derivation trees after expansion. The one for \texttt{\{X/empty\}} is depicted in the middle of Figure \ref{pic:tree1}.
Note that for this tree no expansion was needed as no nodes with open lists have new term matching goals.
\end{example}

\textbf{Compacting and pruning coinductive trees.}
To minimise the amount of used memory and to avoid unnecessary copying of nodes, the implementation contains various mechanisms to remove nodes from the trees in the internal data
structures -- these nodes are guaranteed to be irrelevant for further derivations and to determine whether the tree contains a success subtree. For each found \emph{mgu} for 
the selected open node, a distinct copy of the coinductive tree is created and the unifier applied, afterwards this new tree is expanded. Multiple generated coinductive trees
represent different branches in the coinductive derivation process. If only one branch is needed, the original tree is reused to avoid an unnecessary copy, and the substitution 
is applied directly to it -- this speeds up non branching clauses.  During the copy process, trees are pruned by removing success and non-succeeding subtrees. Also,
chains of single or-nodes are shortened by removing intermediate and-nodes with empty open leaves. These optimisations correspond to the trimming of stack
frames (e.g. tail call optimisation) -- a technique employed in PROLOG.

\textbf{Ordering of solutions.}
Due to the nature of pre-emptive threading, CPU cores working at different speeds, memory access having varying latencies, and trees requiring different computational 
effort; the order of returned trees in the work queue is not deterministic and, therefore, it will change if more than one worker thread is involved. Using the substitution
length of all the substitutions in the derivation chain as priority ranking, we gain an enumeration order even for a potentially infinite lazy derivation processes and the 
implementation can report solutions in this order.

\begin{example}
While an infinite number of coinductive trees can in principle be produced for the goal $\mathtt{btree(X)}$, the algorithm returns the solutions
for the first five success trees in finite time in the following order:

\noindent  $\mathtt{btree(empty)}$ with \texttt{X/empty}\\
$\mathtt{btree(tree(empty,0,empty))}$ with \texttt{X/tree(L1,B1,R1)}, \texttt{L1/empty}, \texttt{B1/0}, \texttt{R1/empty}.\\
$\mathtt{btree(tree(empty,1,empty))}$ with \texttt{X/tree(L1,B1,R1)}, \texttt{L1/empty}, \texttt{B1/1}, \texttt{R1/empty}.\\
$\mathtt{btree(tree(tree(empty,0,empty),0,empty))}$ with \texttt{X/tree(L1,B1,R1)}, \texttt{L1/tree(L2,B2,R2)}, \texttt{L2/empty}, \texttt{B2/0}, \texttt{R2/empty},
\texttt{B1/0}, \texttt{R1/empty}.\\
$\mathtt{btree(tree(empty,0,tree(empty,0,empty)))}$ with \texttt{X/tree(L1,B1,R1)}, \texttt{L1/empty}, \texttt{B1/0}, \texttt{R1/tree(L2,B2,R2)},
\texttt{L2/empty}, \texttt{B2/0}, \texttt{R2/empty}.\\
\end{example}

This is implemented by buffering success trees in a priority queue. If it is determined that the work queue and all workers only hold derivation trees with the same or
higher number of derivation steps, the solutions up to that number are returned from the solution queue. If the work queue is empty and all worker threads are idling,
then all the solutions in the queue are returned and the program exits. Note that the program may continue (co)recursively for indefinitely long.

A PROLOG-like system with deterministic depth-search would produce solutions \texttt{btree(empty)}, \texttt{btree(tree(empty,0,empty))},\linebreak
\texttt{btree(tree(empty,0,tree(empty,0,empty)))}, but not e.g. \texttt{btree(tree(empty,1,empty))}. Thereby, it does not generate the same set of solutions
even if run indefinitely, and does not discover some of the solutions that CoALP does for the \texttt{BinaryTree} program.

\subsection{Case Study: Binary Tree}

\begin{figure}
\begin{center}

\begin{tikzpicture}

\draw (0,0) node {

\begin{tikzpicture}[scale=.75]
	\begin{axis}[xlabel= threads (t) and expand threads (e), xlabel style={at={(0.5,-.2)},anchor=north},
		ylabel=speedup,legend style={at={(0.9,0.1)},anchor=south east},
		xticklabels={1t,2t,3t,4t,5t,6t,6t+1e,6t+2e,6t+3e,6t+4e,6t+5e,6t+6e},xtick={1,...,12}, x tick label style={rotate=90,anchor=east}]
		
	\addplot  coordinates {
		(1,1)	(2,1.87755102040816)	(3,2.3741935483871)	(4,3.2)	(5,3.68)	(6,4.08888888888889)	(7,4.32941176470588)	(8,4.04395604395604)	(9,4.04395604395604)	(10,4.13483146067416)	(11,4)	(12,4.25433526011561)
	};
	
	\addplot coordinates {
(1,1)	(2,1.9051724137931)	(3,2.42857142857143)	(4,3.15714285714286)	(5,3.81034482758621)	(6,4.16981132075472)		(7,4.65263157894737)	(8,4.70212765957447)	(9,4.42)	(10,4.75268817204301)	(11,4.70212765957447)	(12,4.70713525026624)

	};

	\addplot  coordinates {
          (1,1)	(2,1.71698113207547)	(3,2.42021276595745)	(4,3.0952380952381)	(5,3.4469696969697)	(6,4.00881057268723)		(7,4.29245283018868)	(8,4.375)	(9,4.3168880455408)	(10,4.26029962546817)	(11,4.27631578947368)	(12,4.3168880455408)

	};
        \addplot  coordinates {
(1,1)	(2,1.76037037037037)	(3,2.26333333333333)	(4,2.79588235294118)	(5,3.16866666666667)	(6,3.96083333333333)		(7,3.89590163934426)	(8,4.32090909090909)	(9,4.20619469026549)	(10,3.89590163934426)	(11,3.86422764227642)	(12,4.32090909090909)

	};
        \addplot  coordinates {
          (1,1)	(2,1.47878787878788)	(3,2.03333333333333)	(4,2.32380952380952)	(5,2.71111111111111)	(6,3.75384615384615)		(7,3.84251968503937)	(8,3.87301587301587)	(9,3.8125)	(10,3.75384615384615)	(11,3.75384615384615)	(12,4)

	};

\addplot coordinates {
  (1,1)	(2,1.79393939393939)	(3,2.368)	(4,2.96)	(5,3.28888888888889)	(6,4.22857142857143)		(7,4.45112781954887)	(8,4.55384615384615)	(9,4.05479452054795)	(10,4.02721088435374)	(11,4.51908396946565)	(12,3.94666666666667)

	};
        \legend{{\tiny 1000 solutions}, {\tiny 2000 solutions}, {\tiny 3000 solutions}, {\tiny 4000 solutions}, {\tiny 5000 solutions}, {\tiny 6000 solutions}}

	\end{axis}
\end{tikzpicture}
};

\draw (7,0) node{
\begin{tikzpicture}[scale=.75]
	\begin{axis}[xlabel=answers,
		ylabel=time,legend style={at={(.9,.5)},anchor=east}]
	\addplot  coordinates {
		(1000,  7.3)
		(2000,  44.2)
		(3000,  45.5)
		(4000,  47.5)
		(5000,  48.8)
		(6000,  59.2)
	};
	\addplot  coordinates {
		(1000,  1.8)
		(2000,  10.6)
		(3000,  11.3)
		(4000,  12.9)
		(5000,  13.3)
		(6000,  14.3)

	};
	\addplot  coordinates {
		(1000,  1.7)
		(2000,  9.3)
		(3000,  10.4)
		(4000,  11)
		(5000,  12.2)
		(6000,  13)
	};
	
	 \legend{{\tiny 1t}, {\tiny 6t}, {\tiny 6t+6e}}

	\end{axis}
\end{tikzpicture}
};

\end{tikzpicture}

\end{center}
\caption{\textbf{Left.} Speedup of the derivation of answers for the query \texttt{btree(X)}, relative to the base case with 1 worker thread, in the \texttt{BinaryTree} program 
for different parameters.
\textbf{Right.} Comparison of the time needed to produce answers with different worker threads.}\label{fig:speedupderivation}
\end{figure}
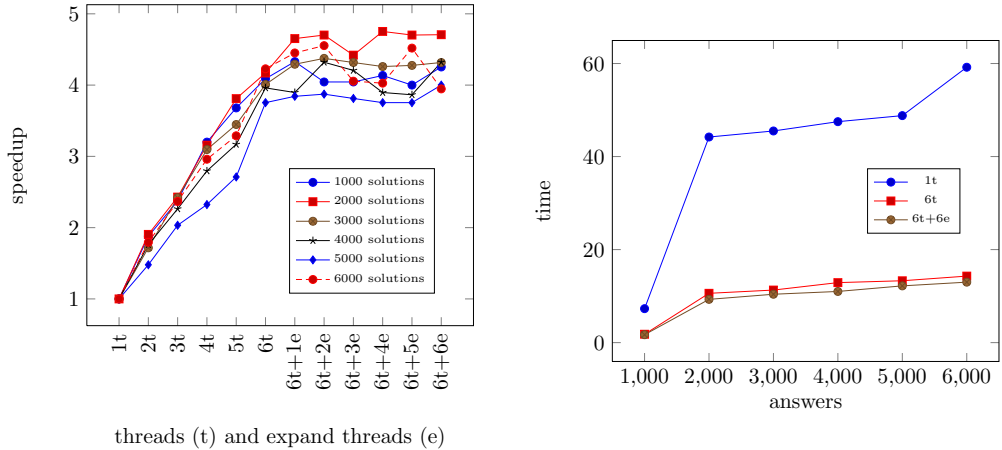

Let us show the benefits of using parallelism to derive answers for the \verb"BinaryTree" program. We focus on deriving different possibilities for substitution in the variable 
\verb"X" in the query \verb"btree(X)". 


Figure~\ref{fig:speedupderivation} shows the speedup when computing different amount of answers for the query \verb"btree(X)" using different parameters. First of all, we can notice
the great benefit of increasing the number of threads to derive answers when the parallel expansion option is not used; in particular, the maximum speedup is $4.23$.
The reason for these good results is simple: the main execution distributes the derivation of different solutions across all the available threads; and as each derivation is independent
from the others, they can be run simultaneously.
On the contrary, the number of expand threads that are used does not impact the performance of our implementation (the maximum speedup is $4.55$ that is not too different from the 
maximum speedup obtained without expand threads).

Note also the considerable difference in the runtime that is necessary to generate 2000 solutions instead of 1000 solutions -- a situation that does not happen in the rest of the cases.
This is due to the fact that the generation of 1000 solutions requires at most 6 derivation steps; on the contrary, the generation of 2000 solutions requires, in some cases, 7 steps. 
Table~\ref{tab:growing} shows that there is a considerable difference to generate all the solutions requiring at most $n$ derivation steps ($2^{n-1}\frac{(2n-2)!}{n!(n-1)!}$ and
$\sum\limits_{i=1}^n 2^{n-1}\frac{(2n-2)!}{n!(n-1)!}$ are the number of solutions that require respectively $n$ and at most $n$ derivation steps)
and the generation of all the solutions requiring at most $n$ derivation steps plus one solution that requires $n+1$ steps -- the time increases up to 4 times.

As a final remark, we can compare the runtime that we obtain here with the Datalog results in Subsection~\ref{subsec:dcsbtg}. 
First-order implementation derives answers faster using the \verb"BinaryTree" program than ground implementation does using the \texttt{BTG} programs. E.g. the first-order derivation 
method with the best configuration obtains $2000$ solutions in $9s$ whereas it takes at least $65s$ in the \texttt{BTG} case.
The lack of variables and references in the ground case means a more intense processing and use of memory space.
The first-order version has fewer rules than the ground case, which makes the use of CPU caches more efficient. 

This result is a good indication that the fibrational coalgebraic approach to first-order parallelism can be viable and efficient. 

\begin{table}
 \centering
 {\footnotesize 
    \begin{tabular}{| l || c | c | }
    \hline
                  & $\sum\limits_{i=1}^n 2^{n-1}\frac{(2n-2)!}{n!(n-1)!}$  solutions&  $1+\sum\limits_{i=1}^n 2^{n-1}\frac{(2n-2)!}{n!(n-1)!}$  solutions   \\ 
    \hline
    \hline
    $n=2$         & 28ms  & 111ms \\
    \hline
    $n=3$         & 178ms & 635ms\\
    \hline
    $n=4$         & 1.18s & 4.68s\\
    \hline
    $n=5$         & 8.92s & 38.83s \\
    \hline
   
\end{tabular}}
\caption{Difference of time when generating all the solutions that require at most $n$ steps and when generating all the solutions requiring at most $n$ derivation steps 
plus one solution that requires $n+1$ steps.}\label{tab:growing} 
\end{table}

\section{Conclusions}\label{sec:evaluation}

We have presented a novel implementation of CoALP, featuring two levels of parallelism: on the level of coinductive trees and on the level of coinductive derivations.
We have traced carefully how the concepts undergo their transformation, from abstract coalgebraic semantics to logic algorithm and then to implementation. 
By doing this, we have exposed the constrictive component of the  coalgebraic semantics~\cite{KP11,KPS12-2}.
The concepts of coinductive trees and coinductive derivations arising from the semantics allow for many approaches of exploiting parallelism. These have  been shown 
to provide practical speedup 
in our experimental implementation. Many improvements are planned to the current implementation, in particular, fine-tuning and-parallel tree transitions and memorisation. 
In the current stage of implementation, no memorisation of tree derivations is done. Much like in PROLOG systems, this however would save doing the same work multiple times
by different workers. A first step to minimise this would be to share subtree structures created during and-parallel derivations. Also, expansion of goals can be cached for
each clause making clause-tree templates not clause, but goal specific. Furthermore, \emph{mgu's} from open nodes could be processed all at once instead of one at a time to 
reduce created copies of trees by finding compatible \emph{mgu's} and applying them at same time. More sophisticated tabling algorithms to memorise computations for each 
clause can also be added. Another project is to apply CoALP to type inference, similar to \cite{AnconaLZ08}.

\bibliographystyle{plain}
\bibliography{KSH13}

\begin{thebibliography}{10}

\bibitem{AmatoLM09}
G.~Amato, J.~Lipton, and R.~McGrail.
\newblock On the algebraic structure of declarative programming languages.
\newblock {\em Theor. Comput. Sci.}, 410(46):4626--4671, 2009.

\bibitem{AnconaLZ08}
D.~Ancona, G.~Lagorio, and E.~Zucca.
\newblock Type inference by coinductive logic programming.
\newblock In {\em TYPES'09}, volume 5497 of {\em LNCS}, pages 1--18, 2009.

\bibitem{BonchiM09}
F.~Bonchi and U.~Montanari.
\newblock Reactive systems, (semi-)saturated semantics and coalgebras on
  presheaves.
\newblock {\em Theor. Comput. Sci.}, 410(41):4044--4066, 2009.

\bibitem{BonchiZ13}
F.~Bonchi and F.~Zanasi.
\newblock Saturated semantics for coalgebraic logic programming.
\newblock In {\em CALCO'13}, volume 8089 of {\em LNCS}, pages 80--94, 2013.

\bibitem{BruniMR01}
R.~Bruni, U.~Montanari, and F.~Rossi.
\newblock An interactive semantics of logic programming.
\newblock {\em Theory Pract. Log. Program.}, 1(6):647--690, 2001.

\bibitem{CominiLM01}
M.~Comini, G.~Levi, and M.~C. Meo.
\newblock A theory of observables for logic programs.
\newblock {\em Inf. Comput.}, 169(1):23--80, 2001.

\bibitem{DKM84}
C.~Dwork, P.C. Kanellakis, and J.C. Mitchell.
\newblock On the sequential nature of unification.
\newblock {\em J.~Logic Prog.}, 1:35--50, 1984.

\bibitem{GLM95}
M.~Gabrielli, G.~Levi, and M.C. Meo.
\newblock Observable behaviors and equivalences of logic programs.
\newblock {\em Inf. Comput.}, 122(1):1--29, 1995.

\bibitem{GuptaC94}
G.~Gupta and V.S. Costa.
\newblock Optimal implementation of and-or parallel {PROLOG}.
\newblock In {\em PARLE'92}, pages 71--92, 1994.

\bibitem{GPACH12}
G.~Gupta, E.~Pontelli, K.~Ali, M.~Carlsson, and M.~Hermenegildo.
\newblock {Parallel Execution of PROLOG Programs: a Survey}.
\newblock {\em ACM Trans.~Comput. Log.}, 23:1--126, 2012.

\bibitem{Hoare:1978:CSP}
C.~A.~R. Hoare.
\newblock Communicating sequential processes.
\newblock {\em Commun. ACM}, 21(8):666--677, August 1978.

\bibitem{Kanellakis88}
P.~C. Kanellakis.
\newblock Logic programming and parallel complexity.
\newblock In {\em Foundations of Deductive Databases and Logic Prog.}, pages
  547--585. Morgan Kaufmann, 1988.

\bibitem{KP96}
Y.~Kinoshita and J.~Power.
\newblock A fibrational semantics for logic programs.
\newblock In {\em Proc. Int. Workshop on Extensions of Logic Programming},
  volume 1050 of {\em LNAI}, pages 177--191, 1996.

\bibitem{KMP10}
E.~Komendantskaya, G.~McCusker, and J.~Power.
\newblock Coalgebraic semantics for parallel derivation strategies in logic
  programming.
\newblock In {\em AMAST'10}, volume 6486 of {\em LNCS}, pages 111--127, 2010.

\bibitem{KP11-2}
E.~Komendantskaya and J.~Power.
\newblock Coalgebraic derivations in logic programming.
\newblock In {\em CSL'11}, LIPIcs, pages 352--366. Schloss Dagstuhl, 2011.

\bibitem{KP11}
E.~Komendantskaya and J.~Power.
\newblock Coalgebraic semantics for derivations in logic programming.
\newblock In {\em CALCO'11}, volume 6859 of {\em LNCS}, pages 268--282.
  Spinger, 2011.

\bibitem{KPS12-2}
E.~Komendantskaya, J.~Power, and M.~Schmidt.
\newblock Coalgebraic logic programming.
\newblock In {\em Journal of Logic and Computation}, 2014.

\bibitem{KPS13}
E.~Komendantskaya, M.~Schmidt, and J.~Heras.
\newblock {CoALP webpage: software and supporting documentation}, 2013.
\newblock http://www.computing.dundee.ac.uk/staff/katya/CoALP/.

\bibitem{Llo88}
J.W. Lloyd.
\newblock {\em Foundations of Logic Programming}.
\newblock Springer-Verlag, 2nd edition, 1987.

\bibitem{PontelliG95}
E.~Pontelli and G.~Gupta.
\newblock On the duality between or-parallelism and and-parallelism in logic
  programming.
\newblock In {\em Euro-Par'95}, volume 966 of {\em LNCS}, pages 43--54, 1995.

\bibitem{S12}
M.~Summerfield.
\newblock {\em Programming in {G}o: Creating Applications for the 21st
  Century}.
\newblock Addison-Wesley, 2012.

\bibitem{UllmanG88}
J.~D. Ullman and A.~Van Gelder.
\newblock Parallel complexity of logical query programs.
\newblock {\em Algorithmica}, 3:5--42, 1988.

\end{thebibliography}

\appendix

\section{Generation of balanced and unbalanced trees}\label{appendix1}

\begin{figure}[h]
\centering
\begin{Verbatim}[frame=single,fontsize=\scriptsize,commandchars=\\\{\},codes={\catcode`$=3\catcode`^=7}]  
BTA
Input: Number of iterations n 

bits = \{0,1\}
trees = \{tree(empty,0,empty), tree(empty,1,empty)\}
new_trees = \{\}
result = \{empty,tree(empty,0,empty), tree(empty,1,empty)\}
i = 0

while (i < n)
  for each t1 in trees
     for each t2 in trees 
        for each b in bits 
            new_trees = new_trees $\cup$ \{tree(t1,b,t2)\}
        end for
    end for
  end for
  i = i + 1
  result = result $\cup$ new_trees
  trees = new_trees
  new_trees=\{\}
end while

return result
\end{Verbatim}
\vspace{-.5cm}
 \caption{{\scriptsize The BTA algorithm generating balanced trees.}}
  \label{fig:BTA}
\end{figure}

\begin{figure}[h]
\centering
\begin{Verbatim}[frame=single,fontsize=\scriptsize,commandchars=\\\{\},codes={\catcode`$=3\catcode`^=7}]  
UTA
Input: Number of iterations n 

bits = \{0,1\}
trees = \{empty, tree(empty,0,empty), tree(empty,1,empty)\}
new_trees = \{\}
i = 0

while (i < n)
  for each t1 in trees
     for each t2 in trees 
        for each b in bits 
            new_trees = new_trees $\cup$ \{tree(t1,b,t2)\}
        end for
    end for
  end for
  i = i + 1
  trees = trees $\cup$ new_trees
  new_trees=\{\}
end while

return trees
\end{Verbatim}
\vspace{-.5cm}
 \caption{{\scriptsize The UTA algorithm generating unbalanced trees.}}
  \label{fig:UTA}
\end{figure}

\section{Execution time of the different programs}\label{appendix:times}

In this section, we show the runtime of the different programs that have appeared 
through the text.

The execution time of the Datalog programs (see end of Section~\ref{subsec:and-or-g} and Figure~\ref{fig:datalog} is given in
Table~\ref{table:datalog}. 

\begin{table}
\begin{center}
    \begin{tabular}{| l || c | c | c | c | c | c |}
    \hline
    \backslashbox{Program}{Threads} & 1 & 2 & 3 & 4 & 5 & 6\\ 
    \hline
    \hline
    datalog 1           & 8s &  6s &  4s &  3s &  3s &  3s\\
                        & 100\% &  75\% & 50\% & 37\% & 37\% & 37\% \\
    \hline
    datalog 2           & 3s &  2s &  2s &  1s &  1s &  1s\\
                        & 100\% &  66\% & 66\% & 33\% & 33\% & 33\% \\
    \hline
    datalog 3           & 50s &  31s &  23s &  19s &  16s &  14s\\
                        & 100\% &  62\% & 46\% & 38\% & 32\% & 28\% \\
    \hline
    datalog 4           & 71s &  44s &  32s &  26s &  23s &  20s\\
                        & 100\% &  61\% & 45\% & 36\% & 32\% & 28\% \\
    \hline
    datalog 5           & 54s & 33s &  23s &  19s &  17s &  14s\\
                        & 100\% &  61\% & 42\% & 35\% & 31\% & 25\% \\
    \hline
    datalog 6           & 53s &  32s &  23s &  18s &  16s &  14s\\
                        & 100\% &  60\% & 43\% & 33\% & 30\% & 26\% \\
    \hline
		\end{tabular}
\end{center}
\caption{Runtime of Datalog programs with different number of threads expanding
the derivation tree.}\label{table:datalog}
\end{table}

Table~\ref{table:bgt} shows the runtime of different versions of the \texttt{BTG} program generated with the \texttt{BTA} algorithm
using different parameters (see Section~\ref{subsec:dcsbtg}); and, Figure~\ref{percentage-bal} shows how the time is reduced when
increasing the number of threads. As can be seen from these results, the and-or derivation trees are generated faster when the \emph{parallel expansion} option 
is activated; namely, the slowest execution using \emph{parallel expansion} is faster than the fastest execution without using it. If we compare the slowest 
execution without using parallel expand with the fastest execution using this option, the runtime is reduced up to $75\%$. Moreover, 
it is worth mentioning that in case the parallel expansion option is not used, increasing the number of threads does not significantly speed up the execution time 
(time is reduced only up to $15\%$); on the contrary, using the parallel expansion option and increasing the number of expansion threads considerably reduces the 
runtime (up to $50\%$).

As can be noticed in Figure~\ref{percentage-bal}, the execution time are greatly reduced when using $3$ threads instead of $2$, but the reduction is smaller every time
 we add a new thread. This can be attributed to the fact that multiple threads have to compete for a limited amount of computing resources such as memory.

\begin{table}
\begin{center}

\hspace{-.5cm}
\begin{minipage}{.7\textwidth}
 
{\scriptsize 
    \begin{tabular}{| l || c | c | c | c | c | c |}
    \hline
    \backslashbox{Program}{Threads} & 1 & 2 & 3 & 4 & 5 & 6    \\ 
    \hline
    \hline
    \texttt{BTG}$_B$(200)           & 8 &  8 & 8 & 8& 7& 7\\
    \hline
    \texttt{BTG}$_B$(400)          & 65 & 60& 58 & 56 & 55  & 55\\
    \hline
    \texttt{BTG}$_B$(600)           & 212  & 193   & 185  & 180& 178 & 176 \\
    \hline
    \texttt{BTG}$_B$(800)           & 444 &406	 &387	 &374	&370	&366 \\ 
    \hline
    \texttt{BTG}$_B$(1000)           & 760  & 685	& 655	& 638	& 	633& 625	\\
    \hline

\end{tabular}}

\vspace{.5cm}

{\scriptsize 
    \begin{tabular}{| l || c | c | c | c | c | c |}
    \hline
    \backslashbox{Program}{E. threads} & 1 & 2 & 3 & 4 & 5 & 6    \\ 
    \hline
    \hline
    \texttt{BTG}$_B$(200)         & 5 & 4 & 3  &  3& 3& 3 \\
    \hline
    \texttt{BTG}$_B$(400)         & 32 &  27 & 21 & 19 & 18  & 17 \\
    \hline
    \texttt{BTG}$_B$(600)           & 100 & 90  & 68 &  61& 57 & 52 \\
    \hline
    \texttt{BTG}$_B$(800)           & 204 &	187 &	141 &	128&	118& 108 \\ 
    \hline
    \texttt{BTG}$_B$(1000)           & 346  & 	321& 	242& 217	& 200	& 	184\\
    \hline
    
\end{tabular}}

\end{minipage}
 \begin{minipage}{.3\textwidth}
 
 \vspace{.2cm}
 \begin{tikzpicture}[scale=.5]
 	\begin{axis}[xlabel=number of leaves,
 		ylabel=time,legend style={at={(0.1,.9)},anchor=north west}]
 	\addplot  coordinates {
 		(3960,8)
 		(10160,65)
 		(16360,212)
 		(22560,444)
 		(28760,760)
 	};
 	\addplot coordinates {
 		(3960,7)
 		(10160,55)
 		(16360,178)
 		(22560,366)
 		(28760,625)
 	};
 	\addplot coordinates {
 		(3960,5)
 		(10160,32)
 		(16360,100)
 		(22560,204)
 		(28760,346)
 	};
 	\addplot  coordinates {
 		(3960,3)
 		(10160,17)
 		(16360,52)
 		(22560,108)
 		(28760,184)
 	};
 	 \legend{{\tiny 1 thread}, {\tiny 6 threads}, {\tiny 1 expand thread}, {\tiny 6 expand threads}}
 
 	\end{axis}
 
%
%
 \end{tikzpicture}

 \end{minipage}

\end{center}
\caption{Runtime (in seconds) of different versions of \texttt{BTG} program with different parameters -- the value of $X$ in \texttt{BTG}$_B(X)$ indicates the number of clauses.
\textbf{Left. Top}. Runtime using different number of threads.
\textbf{Left. Bottom}. Runtime using different number of expansion threads. \textbf{Right}. Comparison of the time of the slowest and fastest executions using threads and
expand threads regarding the number of leaves. The query for all the programs is \texttt{btree(X)} and the runtime indicate the time (in seconds) needed to obtain all the possible answers. 
}\label{table:bgt}
\end{table}

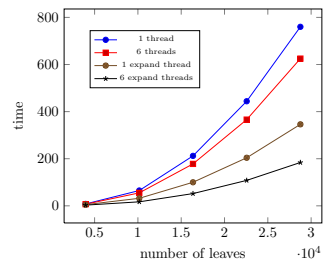
\begin{figure}
\centering
\begin{tikzpicture}
\draw (0,0) node{
\begin{tikzpicture}[scale=.75]
	\begin{axis}[xlabel=threads,
		ylabel=percentage relative to 1 thread]
	\addplot  coordinates {
		(1,100)
		(2,99)
		(3,98)
		(4,97)
		(5,88)
		(6,87)
	};
	\addplot  coordinates {
		(1,100)
		(2,92)
		(3,89)
		(4,86)
		(5,84)
		(6,84)
	};
	\addplot  coordinates {
		(1,100)
		(2,91)
		(3,87)
		(4,84)
		(5,83)
		(6,83)
	};
	\addplot  coordinates {
		(1,100)
		(2,91)
		(3,87)
		(4,84)
		(5,83)
		(6,82)
	};
	\addplot coordinates {
		(1,100)
		(2,90)
		(3,86)
		(4,83)
		(5,83)
		(6,82)
	};
	 \legend{{\tiny BTG$_B$(200)}, {\tiny BTG$_B$(400)}, {\tiny BTG$_B$(600)}, {\tiny BTG$_B$(800)}, {\tiny BTG$_B$(1000)}}

	\end{axis}
\end{tikzpicture}};

\draw (7,0) node{
\begin{tikzpicture}[scale=.75]
	\begin{axis}[
		xlabel=threads,
		ylabel=percentage relative to 1 thread]
	\addplot  coordinates {
		(1,100)
		(2,80)
		(3,60)
		(4,60)
		(5,60)
		(6,60)
	};
	\addplot coordinates {
		(1,100)
		(2,84)
		(3,65)
		(4,59)
		(5,56)
		(6,53)
	};
	\addplot  coordinates {
		(1,100)
		(2,90)
		(3,68)
		(4,61)
		(5,57)
		(6,52)
	};
	\addplot  coordinates {
		(1,100)
		(2,91)
		(3,69)
		(4,62)
		(5,57)
		(6,52)
	};
	\addplot coordinates {
		(1,100)
		(2,92)
		(3,69)
		(4,62)
		(5,57)
		(6,53)
	};
	 \legend{{\tiny BTG$_B$(200)}, {\tiny BTG$_B$(400)}, {\tiny BTG$_B$(600)}, {\tiny BTG$_B$(800)}, {\tiny BTG$_B$(1000)}}

	\end{axis}	
\end{tikzpicture}};
\end{tikzpicture}

\caption{Reduction percentages relative to the 1 thread case on the \texttt{BTG}$_B$(X) program. \textbf{Left.} Reduction relative to 1 thread. 
\textbf{Right.} Reduction relative to 1 expand thread.}\label{percentage-bal}
\end{figure}
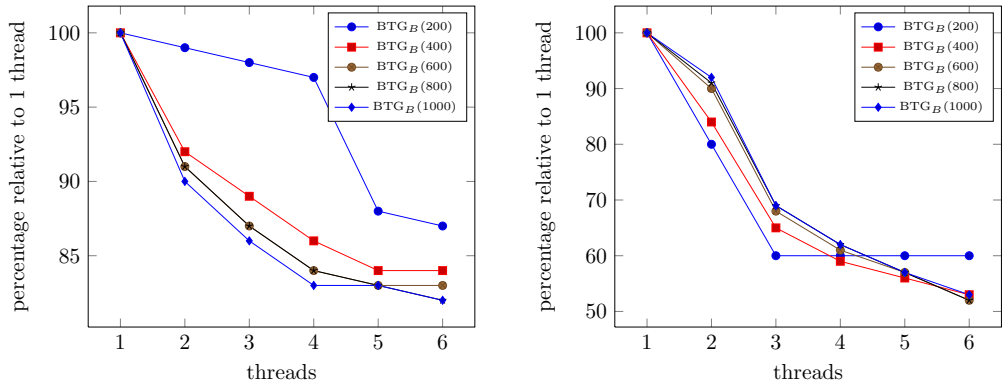

Table~\ref{table:bgt-unb} shows the runtime of different versions of the \texttt{BTG} program generated with the \texttt{UTA} algorithm
using different parameters (see Section~\ref{subsec:dcsbtg}); and, Figure~\ref{percentage-unbal} shows how the time is reduced when
increasing the number of threads.

\begin{table}
\begin{center}

 \begin{minipage}{.65\textwidth}

{\scriptsize 
    \begin{tabular}{| l || c | c | c | c | c | c |}
    \hline
    \backslashbox{Program}{Threads} & 1 & 2 & 3 & 4 & 5 & 6    \\ 
    \hline
    \hline
    \texttt{BTG}$_U$(500)           &  5&	4	&4	&4&	4	&4 \\
    \hline
    \texttt{BTG}$_U$(1000)          &  18&	17	&16	&15	&15	&15  \\
    \hline
    \texttt{BTG}$_U$(1500)            &  50&	44&	42	&40	&40&	39 \\
    \hline
    \texttt{BTG}$_U$(2000)            & 104&	91&	86&	83&	82&  81\\
    \hline
    \texttt{BTG}$_U$(2500)           & 196  & 173  &  163&  158& 155 & 154\\
    \hline

\end{tabular}}

\vspace{.5cm}

{\scriptsize 
    \begin{tabular}{| l || c | c | c | c | c | c |}
    \hline
    \backslashbox{Program}{E. threads} & 1 & 2 & 3 & 4 & 5 & 6    \\ 
    \hline
    \hline
    \texttt{BTG}$_U$(500)           & 4  &  4 & 4 & 4 & 4 & 4\\
    \hline
    \texttt{BTG}$_U$(1000)          & 14  &  14&   14& 14  & 14   & 14 \\
    \hline
    \texttt{BTG}$_U$(1500)            &36   & 36  & 34 & 34 & 34 &34 \\
    \hline
    \texttt{BTG}$_U$(2000)            &  72 &  71 & 65 & 65 & 65  & 65 \\
    \hline
    \texttt{BTG}$_U$(2500)           & 130  & 128  & 112 & 111 & 110  & 110\\
    \hline
    
\end{tabular}}

 \end{minipage}
 \begin{minipage}{.25\textwidth}
 
 \vspace{.2cm}
 \begin{tikzpicture}[scale=.5]
 \begin{axis}[xlabel=number of leaves,
 		ylabel=time,legend style={at={(0.1,.9)},anchor=north west}]
 	\addplot  coordinates {
 		(6656, 5)
 		(13352, 18)
 		(20590, 50)
 		(28518, 104)
 		(35248, 196)
 	};
 	\addplot  coordinates {
 		(6656, 4)
 		(13352, 15)
 		(20590, 39)
 		(28518, 81)
 		(35248, 154)
 	};
 	\addplot  coordinates {
 		(6656, 4)
 		(13352, 14)
 		(20590, 36)
 		(28518, 72)
 		(35248, 130)
 	};
 	\addplot  coordinates {
 		(6656, 4)
 		(13352, 14)
 		(20590, 34)
 		(28518, 65)
 		(35248, 110)
 	};
 	
 	 \legend{{\tiny 1 thread}, {\tiny 6 threads}, {\tiny 1 expand thread}, {\tiny 6 expand threads}}
 
 	\end{axis}
%
%
%
%
%
 \end{tikzpicture}
 
 \vspace{.5cm}
 
 \end{minipage}

\end{center}
\caption{Runtime (in seconds) of different versions of \texttt{BTG} program with different parameters -- the value of $X$ in \texttt{BTG}$_U(X)$ indicates the number of clauses. \textbf{Left. Top}. Runtime using different number of threads
\textbf{Left. Bottom}. Runtime using different number of expansion threads. \textbf{Right.} Comparison of the time of the slowest and fastest executions using 
threads and expand threads regarding the number of leaves. The query for all the programs is \texttt{btree(X)} and the
runtime indicate the time needed to obtain all the possible answers.}\label{table:bgt-unb}
\end{table}

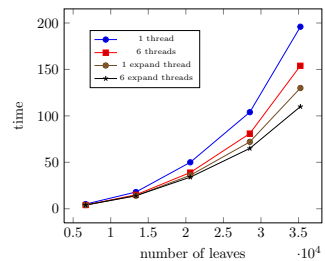
\begin{figure}
\centering
\begin{tikzpicture}
\draw (0,0) node{
\begin{tikzpicture}[scale=.75]
	\begin{axis}[xlabel=threads,
		ylabel=percentage relative to 1 thread]
	\addplot  coordinates {
		(1,100)
		(2,80)
		(3,80)
		(4,80)
		(5,80)
		(6,80)
	};
	\addplot  coordinates {
		(1,100)
		(2,94)
		(3,88)
		(4,83)
		(5,83)
		(6,83)
	};
	\addplot coordinates {
		(1,100)
		(2,88)
		(3,84)
		(4,80)
		(5,80)
		(6,78)
	};
	\addplot  coordinates {
		(1,100)
		(2,87)
		(3,82)
		(4,79)
		(5,78)
		(6,77)
	};
	\addplot coordinates {
		(1,100)
		(2,88)
		(3,83)
		(4,80)
		(5,79)
		(6,78)
	};
	 \legend{{\tiny BTG$_U$(500)}, {\tiny BTG$_U$(1000)}, {\tiny BTG$_U$(1500)}, {\tiny BTG$_U$(2000)}, {\tiny BTG$_U$(2500)}}

	\end{axis}
\end{tikzpicture}};

\draw (7,0) node{
\begin{tikzpicture}[scale=.75]
	\begin{axis}[
		xlabel=threads,
		ylabel=percentage relative to 1 thread,legend style={at={(0.05,0.05)},anchor=south west}]
	\addplot  coordinates {
		(1,100)
		(2,100)
		(3,100)
		(4,100)
		(5,100)
		(6,100)
	};
	\addplot coordinates {
		(1,100)
		(2,100)
		(3,100)
		(4,100)
		(5,100)
		(6,100)
	};
	\addplot coordinates {
		(1,100)
		(2,100)
		(3,94)
		(4,94)
		(5,94)
		(6,94)
	};
	\addplot coordinates {
		(1,100)
		(2,98)
		(3,90)
		(4,90)
		(5,90)
		(6,90)
	};
	\addplot coordinates {
		(1,100)
		(2,98)
		(3,86)
		(4,85)
		(5,84)
		(6,84)
	};
	 \legend{{\tiny BTG$_U$(500)}, {\tiny BTG$_U$(1000)}, {\tiny BTG$_U$(1500)}, {\tiny BTG$_U$(2000)}, {\tiny BTG$_U$(2500)}}

	\end{axis}	
\end{tikzpicture}};
\end{tikzpicture}

\caption{Reduction percentages relative to the 1 thread case on the \texttt{BTG}$_U$(X) program. \textbf{Left.} Reduction relative to 1 thread. 
\textbf{Right.} Reduction relative to 1 expand thread.}\label{percentage-unbal}
\end{figure}

The runtime and the reduction percentages for Example~\ref{ex:ttree} are shown respectively in Table~\ref{table:ttree} and Figure~\ref{percentage-ttree}.

\begin{table}
\begin{center}

\begin{minipage}{.65\textwidth}
 
{\scriptsize 
    \begin{tabular}{| l || c | c | c | c | c | c |}
    \hline
    \backslashbox{i}{Threads} & 1 & 2 & 3 & 4 & 5 & 6    \\ 
    \hline
    \hline
    \texttt{ttree(s$^{11}$(0))}         & 3.7 &	3.3 &3.1	 &3.1	 &3.0	 &3.1   \\
    \hline
    \texttt{ttree(s$^{12}$(0))}          & 11.7  &10.6	 &10.0	 &9.9	 &9.7	 &9.6   \\
    \hline
    \texttt{ttree(s$^{13}$(0))}           &37.1  &33.6	 & 31.6	 & 30.9	 & 30.2	 &30.1   \\
    \hline
    \texttt{ttree(s$^{14}$(0))}           & 96.6  & 83.4	 &78.4	 & 76.8	 & 74.4	 & 73.7   \\
    \hline

\end{tabular}}

\vspace{.5cm}

{\scriptsize 
    \begin{tabular}{| l || c | c | c | c | c | c |}
    \hline
    \backslashbox{i}{E. threads} & 1 & 2 & 3 & 4 & 5 & 6    \\ 
    \hline
    \hline
    \texttt{ttree(s$^{11}$(0))}         & 2.2 &	1.9 &	1.9 &	1.8 & 1.7	 &1.7   \\
    \hline
    \texttt{ttree(s$^{12}$(0))}          & 5.2 &	4.2 & 3.9	 &3.5	 &3.3	 &3.2   \\
    \hline
    \texttt{ttree(s$^{13}$(0))}           & 15.7  &11.2	 &10.1	 &9.1	 &8.4	 &8.2   \\
    \hline
    \texttt{ttree(s$^{14}$(0))}           & 40.3 &	31.5 &	 29.8	 &	26.0 & 24.2 & 23.7   \\
    \hline
    
\end{tabular}}

\end{minipage}
\begin{minipage}{.25\textwidth}

\vspace{.2cm}
\begin{tikzpicture}[scale=.5]
	\begin{axis}[xlabel=i,
		ylabel=time,legend style={at={(0.1,.9)},anchor=north west}]
	\addplot coordinates {
		(11,3.7)
		(12,11.7)
		(13,37.1)
		(14,96.6)
	};
	\addplot coordinates {
		(11,3.1)
		(12,9.6)
		(13,30.1)
		(14,73.6)
	};
	\addplot coordinates {
		(11,2.2)
		(12,5.2)
		(13,15.7)
		(14,40.3)
	};
	\addplot coordinates {
		(11,1.7)
		(12,3.2)
		(13,8.2)
		(14,23.7)
	};
	
	 \legend{{\tiny 1 thread}, {\tiny 6 threads}, {\tiny 1 expand thread}, {\tiny 6 expand threads}}

	\end{axis}
\end{tikzpicture}

\vspace{.5cm}

\end{minipage}

\end{center}
\caption{Runtime (in seconds) of the construction of the coinductive tree of \texttt{ttree(s$^i$(0))} for different values of $i$ and 
different parameters. \textbf{Left. Top}. Runtime using different number of threads
\textbf{Left. Bottom}. Runtime using different number of expansion threads. \textbf{Right.} Comparison of the time of the slowest and fastest executions using 
threads and expand threads.}\label{table:ttree}
\end{table}

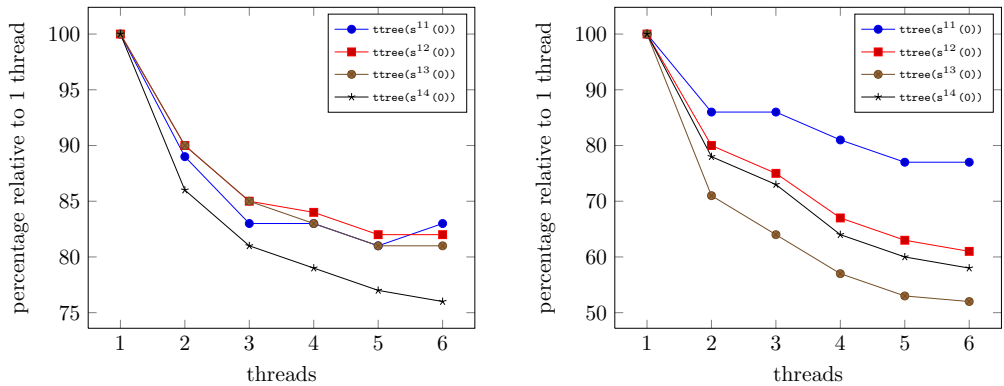
\begin{figure}
\centering
\begin{tikzpicture}
\draw (0,0) node{
\begin{tikzpicture}[scale=.75]
	\begin{axis}[xlabel=threads,
		ylabel=percentage relative to 1 thread]
	\addplot coordinates {
		(1,100)
		(2,89)
		(3,83)
		(4,83)
		(5,81)
		(6,83)
	};
	\addplot coordinates {
		(1,100)
		(2,90)
		(3,85)
		(4,84)
		(5,82)
		(6,82)
	};
	\addplot coordinates {
		(1,100)
		(2,90)
		(3,85)
		(4,83)
		(5,81)
		(6,81)
	};
	\addplot coordinates {
		(1,100)
		(2,86)
		(3,81)
		(4,79)
		(5,77)
		(6,76)
	};
	 \legend{{\tiny \texttt{ttree(s$^{11}$(0))} }, {\tiny \texttt{ttree(s$^{12}$(0))} }, {\tiny \texttt{ttree(s$^{13}$(0))} }, {\tiny \texttt{ttree(s$^{14}$(0))} }}

	\end{axis}
\end{tikzpicture}};

\draw (7,0) node{
\begin{tikzpicture}[scale=.75]
	\begin{axis}[
		xlabel=threads,
		ylabel=percentage relative to 1 thread]
	\addplot coordinates {
		(1,100)
		(2,86)
		(3,86)
		(4,81)
		(5,77)
		(6,77)
	};
	\addplot coordinates {
		(1,100)
		(2,80)
		(3,75)
		(4,67)
		(5,63)
		(6,61)
	};
	\addplot coordinates {
		(1,100)
		(2,71)
		(3,64)
		(4,57)
		(5,53)
		(6,52)
	};
	\addplot coordinates {
		(1,100)
		(2,78)
		(3,73)
		(4,64)
		(5,60)
		(6,58)
	};
	 \legend{{\tiny \texttt{ttree(s$^{11}$(0))} }, {\tiny \texttt{ttree(s$^{12}$(0))} }, {\tiny \texttt{ttree(s$^{13}$(0))} }, {\tiny \texttt{ttree(s$^{14}$(0))} }}	
	\end{axis}	
\end{tikzpicture}};
\end{tikzpicture}

\caption{Reduction percentages relative to the 1 thread case on the construction of the coinductive tree for different \texttt{ttree(s$^i$(0))} in the 
Ttree program. \textbf{Left.} Reduction relative to 1 thread.  \textbf{Right.} Reduction relative to 1 expand thread.}\label{percentage-ttree}
\end{figure}

The runtime to derive different number of answers for the \texttt{BTG} program using different parameters is shown in Figure~\ref{table:btg-generation}

\begin{table}
 
\centering
 
{\footnotesize 
    \begin{tabular}{| l || c | c | c | c | c | c |}
    \hline
    \backslashbox{solutions}{Threads} & 1 & 2 & 3 & 4 & 5 & 6    \\ 
    \hline
    \hline
    1000         &  7.3	&3.9	 &3.1	 &2.3	 &2	 &1.8   \\
    \hline
    2000         & 44.2 & 23.2	 & 18.2	 &14	 & 11.6	 & 10.6   \\
    \hline
    3000         &45.5  &26.5&	18.8	&14.7	&13.2	&11.3\\
    \hline
    4000         & 47.5 &27.2	&21.1&	17.3&	15.6	&12.9\\
    \hline
    5000         & 48.8 &33.2&	24.5&	21.8&	18.3&	13.3   \\
    \hline
    6000         & 59.2 &33.6	&25.4	&20.1	&18.2&	14.3   \\
    \hline

\end{tabular}}

\vspace{.5cm}

{\footnotesize 
    \begin{tabular}{| l || c | c | c | c | c | c |}
    \hline
    \backslashbox{solutions}{E. threads} & 1 & 2 & 3 & 4 & 5 & 6    \\ 
        \hline
    \hline
    1000         & 1.7&	1.8&	1.8&	1.7	&1.8	&1.73   \\
    \hline
    2000         & 9.5&	9.4&	10&	9.3&	9.4	&9.3  \\
    \hline
    3000         & 10.6 &	10.4 &	10.5 &	10.6	 &10.6	 &10.5  \\
    \hline
    4000         & 12.2&11&	11.3&	12.2&	12.3&	11   \\
    \hline
    5000         &12.7&	12.6&	12.8&	13&	13&	12.2  \\
    \hline
    6000         & 13.3&	13	&14.6	&14.7&	13.1&	15\\
    \hline
\end{tabular}}

\caption{Runtime to derive different number of answers for the \texttt{BTG} program using different parameters.}\label{table:btg-generation}
\end{table}

\newpage
\section{Proofs}\label{appendix:proofs}

\begin{proof}[Lemma~\ref{lemma:and-or}]\label{proof:and-or}
Let $P$ be a ground logic program and let $G=\gets A$ be an atomic ground goal. Let $T$ be the and-or parallel tree for $P$ and $G$ by definition \ref{df:andortree}. Let $T'$ be the tree for $P$ and $G$ constructed by definition \ref{cons:and-or}.
Since there are no variables within the goal or program every  \emph{mgu's} found for the unification of two terms must be the empty set of substitutions.

The root of $T$ and $T'$ is both the and-node with atom $G$ by their respective definitions. 
Given an and-node with atom $A$ in $T$ and the corresponding and-node with atom $A$ in $T'$ then there are exactly $m>0$ distinct clauses 
$C_1, \ldots , C_m$ in $P$ with $A = B^i\theta_i = \ldots = B^m\theta_m = B^i = \ldots = B^m$ by definition \ref{df:andortree} and since the \emph{mgu's} can contain no substitutions. For every such clause $C_i = B_i \gets B^i_1, \ldots ,B^i_n$ 
the and-node $A$ in $T$ therefore contains an child or-node with the child and-nodes $B^i_1, \ldots ,B^i_n$ by definition \ref{df:andortree}.
The clause-tree for $C_i$ contains the or-node with children and-nodes $B^i_1, \ldots ,B^i_n$ by definition \ref{def:clause-tree}. 
$A$ in $T'$ is by construction either the root and-node or a copy of an and-node from a clause trsee. If $A$ in $T'$ is the root node it
would have contained by construction \ref{cons:and-or} a reference in the open list to all the corresponding clause trees for $C_1, \ldots , C_m$ since all are unifiable clause trees and are referenced by construction. If it is not the root node then $A$ in $T'$ is the copy of an and-node from a clause tree which by definition of a clause tree contains a reference in the open list to all the unifiable clause trees. Since the correspond \emph{mgu's} can contain no substitutions these exactly correspond again to $C_1, \ldots , C_m$. 
By construction \ref{cons:and-or} all these references to $C_1, \ldots , C_m$ have been removed from the open list of the and-node $A$ in $T'$ and for the reference to the clause tree $C_i$ an or-node with child and-nodes $B^i_1, \ldots ,B^i_n$ has been added. Thereby both $A$ and $A'$ contain the same or-node children with and-node children $B^i_1, \ldots ,B^i_n$. 
\end{proof}

\begin{proof}[Lemma~\ref{lemma:cotreecons}]\label{proof:cotreecons}
Let $P$ be a ground logic program and let $G=\gets A$ be an atomic goal. Let $T$ be the coinductive tree for $P$ and $G$ by definition \ref{df:coindt}. 
Let $T'$ be the tree for $P$ and $G$ constructed by definition \ref{cons:cotree}.

The root of $T$ and $T'$ is both the and-node with atom $G$ by their respective definitions. 
Given an and-node with atom $A$ in $T$ and the corresponding and-node with atom $A$ in $T'$ then there are exactly $m>0$ distinct clauses 
$C_1, \ldots , C_m$ in $P$ (a clause $C_i$ has the form $B_i
  \gets B^i_1, \ldots , B^i_{n_i}$, for some $n_i$) such that $A = B_1\theta_1 =
  ... = B_m\theta_m$, for mgus $\theta_1, \ldots , \theta_m$, by definition \ref{df:coindt}. For every such clause $C_i$ 
the and-node $A$ in $T$ contains an child or-node with the child and-nodes $B^i_1\theta_i, \ldots ,B^i_{n_i}\theta_i$ by definition \ref{df:coindt}. The clause-tree for $C_i$ contains the or-node with children and-nodes $B_i, \ldots ,B^i_n$ by definition \ref{def:clause-tree}. $A$ in $T'$ is by construction either the root and-node or a copy of an and-node from a clause tree with possible substitutions $\theta_a$ applied. If $A$ in $T'$ is the root node it would have contained by construction \ref{cons:cotree} a reference in the open list to at least all the clause trees for $C_1, \ldots , C_m$ with matching root node since these are a subset of the unifiable clause trees root nodes. If it is not the root node then $A$ in $T'$ is the copy of an and-node from a clause tree with possibly some substitutions $\theta_a$ applied. Lets denote this original node atom by $A^*$ with $A=A^*\theta_a$. By definition of a clause tree the node for $A^*$ contains a reference in the open list to all the unifiable clause 
trees roots. Since $C_1, \ldots , C_m$ are matching $A$ they are also unifiable with $A$ and therefore also unfiable with $A^*$. Therefore, they are included in the open list references for $A$ and $A^*$. By construction \ref{cons:cotree} all these references have been removed from the open list and for the referenced clause tree roots that match have been added to the tree. Since references to $C_1, \ldots , C_m$ are contained in the open lists and are exactly the matching ones for each clause tree $C^i$ an or-node with child and-nodes $B^i_1, \ldots ,B^i_n$ has been added with the \emph{mgu} $\theta_i$ applied which then resulted in the and-nodes $B^i_1\theta_i, \ldots ,B^i_{n_i}\theta_i$. Thereby both $A$ and $A'$ contain the same or-node children with and-node children $B^i_1\theta_i, \ldots ,B^i_{n_i}\theta_i$. 
 
\end{proof}

\end{document}